%% file: arxiv.tex
\pdfoutput=1
\documentclass[11pt]{article}

\usepackage[x11names]{xcolor}

\usepackage[utf8]{inputenc}
\usepackage[T1]{fontenc}
\usepackage[sc]{mathpazo}
\usepackage{fullpage}
\usepackage{authblk}

\usepackage{amsmath,amssymb,amsthm}
\usepackage[colorlinks=true,citecolor=Blue4,linkcolor=Red3]{hyperref}
\usepackage[round,sort]{natbib}

\input{macro}
\newcommand{\acks}[1]{\section*{Acknowledgments}#1}

\title{Some Inapproximability Results of MAP Inference and Exponentiated Determinantal Point Processes}
\author{Naoto Ohsaka\thanks{\href{mailto:ohsaka\_naoto@cyberagent.co.jp}{\texttt{ohsaka\_naoto@cyberagent.co.jp}}}}
\affil{CyberAgent, Inc.}
\date{\today}

\begin{document}

\maketitle

\input{main}
\appendix
\input{app}

\bibliographystyle{abbrvnat}
\bibliography{ref}

\end{document}

%% file: macro.tex
\usepackage{booktabs,enumitem}
\usepackage[normalem]{ulem}
\usepackage{framed,fancybox,ascmac}

\usepackage{comment}
\usepackage{url}

\usepackage[colorinlistoftodos,prependcaption]{todonotes}
\usepackage{xspace}

\usepackage{bm}

\usepackage[nameinlink]{cleveref}
\usepackage{nameref}

\usepackage{tcolorbox}

\usepackage{lineno}

\newtheorem{theorem}{Theorem}[section]
\newtheorem*{theorem*}{Theorem}
\newtheorem{lemma}[theorem]{Lemma}
\newtheorem{corollary}[theorem]{Corollary}
\newtheorem{remark}[theorem]{Remark}
\newtheorem{claim}[theorem]{Claim}

\newtheorem{observation}[theorem]{Observation}
\theoremstyle{definition}
\newtheorem{definition}[theorem]{Definition}

\crefname{theorem}{Theorem}{Theorems}
\crefname{lemma}{Lemma}{Lemmas}
\crefname{claim}{Claim}{Claims}
\crefname{remark}{Remark}{Remarks}
\crefname{observation}{Observation}{Observations}
\crefname{corollary}{Corollary}{Corollaries}
\crefname{appendix}{Appendix}{Appendices}
\crefname{section}{Section}{Sections}
\crefname{equation}{Eq.}{Eqs.}
\crefname{figure}{Figure}{Figures}
\crefname{table}{Table}{Tables}

\newcommand{\prb}[1]{\textup{\textsc{#1}}\xspace}
\newcommand{\volmax}{\prb{VolMax}}
\newcommand{\subdetmax}{\prb{DetMax}}
\newcommand{\gapsubdetmax}[1]{#1\prb{-Gap-DetMax}}
\newcommand{\logdetmax}{\prb{LogDetMax}}

\renewcommand{\vec}[1]{\mathbf{\bm{#1}}}
\newcommand{\mat}[1]{\mathbf{\bm{#1}}}
\newcommand{\GG}{\mathfrak{G}}
\newcommand{\ZZ}{\mathcal{Z}}
\newcommand{\bigO}{\mathcal{O}}

\newcommand{\ccc}{\lambda_{{c}}}
\newcommand{\sss}{\lambda_{{s}}}

\DeclareMathOperator{\val}{\mathop{\mathrm{val}}}
\DeclareMathOperator{\vol}{\mathop{\mathrm{vol}}}
\DeclareMathOperator{\maxvol}{\mathop{\mathrm{maxvol}}}
\DeclareMathOperator{\maxdet}{\mathop{\mathrm{maxdet}}}
\DeclareMathOperator{\proj}{\mathop{\mathrm{proj}}}
\DeclareMathOperator{\rep}{\mathop{\mathrm{rep}}}
\DeclareMathOperator{\dis}{\mathop{\mathrm{d}}}
\DeclareMathOperator{\sgn}{\mathop{\mathrm{sgn}}}

\let\det\relax\DeclareMathOperator{\det}{\mathop{\mathrm{det}}}

\newcommand{\rme}{\mathrm{e}}

\newcommand{\bbQ}{\mathbb{Q}}
\newcommand{\bbR}{\mathbb{R}}

\newcommand{\ALG}{\textsc{alg}}

\newcommand{\cP}{\textup{\textsf{P}}\xspace}
\newcommand{\RP}{\textup{\textsf{RP}}\xspace}
\newcommand{\UP}{\textsf{UP}\xspace}
\newcommand{\ModkP}[1]{\textsf{Mod}$_{#1}$\textsf{P}\xspace}
\newcommand{\NP}{\textup{\textsf{NP}}\xspace}
\newcommand{\APX}{\textup{\textsf{APX}}\xspace}
\newcommand{\shP}{$\#$\textsf{P}\xspace}

%% file: main.tex
\begin{abstract}
We study the computational complexity of two hard problems on determinantal point processes (DPPs).
One is \emph{maximum a posteriori (MAP) inference}, i.e.,
to find a principal submatrix having the maximum determinant.
The other is probabilistic inference on \emph{exponentiated DPPs (E-DPPs)},
which can sharpen or weaken the diversity preference of DPPs with an exponent parameter $p$.
We present several complexity-theoretic hardness results that
explain the difficulty
in approximating MAP inference and the normalizing constant for E-DPPs.
We first prove that unconstrained MAP inference for
an $n \times n$ matrix is \NP-hard to approximate
within a factor of $2^{\beta n}$, where $\beta = 10^{-10^{13}} $.
This result improves upon the best-known inapproximability factor of $(\frac{9}{8}-\epsilon)$,
and rules out the existence of any polynomial-factor approximation algorithm assuming \cP~$\neq$~\NP.
We then show that
log-determinant maximization is \NP-hard to approximate within a factor of $\frac{5}{4}$ for the unconstrained case and
within a factor of $1+10^{-10^{13}}$ for the size-constrained monotone case.
In particular, log-determinant maximization does not admit a polynomial-time approximation scheme unless \cP~$=$~\NP.
As a corollary of the first result, we demonstrate that
the normalizing constant for E-DPPs of any (fixed) constant exponent $p \geq \beta^{-1} = 10^{10^{13}}$
is \NP-hard to approximate within a factor of $2^{\beta pn}$,
which is in contrast to the case of $p \leq 1$ admitting a fully polynomial-time randomized approximation scheme.
\end{abstract}

\section{Introduction}
Selecting a small set of ``diverse'' items from large data is an essential task in artificial intelligence.
\emph{Determinantal point processes (DPPs)} provide a probabilistic model on a discrete set that captures the notion of diversity using the matrix determinant \citep{macchi1975coincidence,borodin2005eynard}.
Suppose we are given $n$ items (e.g., images or documents) associated with feature vectors
$\{ \bm{\phi}_i \}_{i \in [n]}$ and
an $n \times n$ Gram matrix $\mat{A}$ such that $A_{i,j} = \langle \bm{\phi}_i, \bm{\phi}_j \rangle$ for all $i,j \in [n]$.
The DPP defined by $\mat{A}$ is a distribution over the power set $2^{[n]}$ such that
the probability of drawing a subset $S \subseteq [n]$ is proportional to $\det(\mat{A}_S)$.
Since $\det(\mat{A}_S)$ is equal to the squared volume of the parallelepiped spanned by $\{\bm{\phi}_i\}_{i \in S}$,
dissimilar items are likely to appear in the selected subsets, which ensures set diversity.
DPPs exhibit fascinating properties that make them suitable for artificial intelligence and machine learning applications; e.g.,
many inference tasks are computationally tractable, including
normalization, marginalization, and sampling, and
efficient learning algorithms have been developed
(see, e.g., the survey of \citealp{kulesza2012determinantal} for details).

The present study aims at analyzing two exceptionally hard problems on DPPs
through the lens of complexity theory---MAP inference and probabilistic inference on exponentiated DPPs.

\paragraph{Unconstrained MAP Inference.}
Seeking the most diverse subset that has the highest probability, i.e., \emph{maximum a posteriori (MAP) inference},
is motivated by numerous applications, e.g.,
document summarization \citep{kulesza2011learning,chao2015large,perez2021multi},
tweet timeline generation \citep{yao2016tweet},
YouTube video recommendation \citep{wilhelm2018practical},
active learning \citep{biyik2019batch}, and
video summarization \citep{gong2014diverse,han2017faster}.
In particular, we focus on \emph{unconstrained} MAP inference,
which is equivalent to finding a principal submatrix with the maximum determinant, i.e., $\max_{S \subseteq [n]} \det(\mat{A}_S) $, and
it has been computationally challenging despite its simplicity.
Typically, the \textsc{Greedy} algorithm is used as a heuristic, whereas
the current best approximation algorithm achieves a factor of $\rme^{n}$ \citep{nikolov2015randomized}.
\citet*{kulesza2012determinantal} have shown that
unconstrained MAP inference is \NP-hard
to approximate within a factor of $(\frac{9}{8}-\epsilon)$,\footnote{
We define approximation factor $\rho$ so that $\rho \geq 1$. See \cref{sec:subdetmax} for the formal definition.}
which is the best-known lower bound.
On the other hand,
\emph{size-constrained} MAP inference (i.e., it must hold that $|S|=k$ for an input $k$) is
\NP-hard to approximate within an exponential factor of $2^{ck}$ for some $c > 0$ \citep{civril2013exponential},
which does \emph{not}, however, directly apply to the unconstrained case.
Closing the gap between the lower bound ($\approx \frac{9}{8}$) and upper bound ($= \rme^{n}$) on unconstrained MAP inference is the first question addressed in this study.

\paragraph{Log-Determinant Maximization.}
In the artificial intelligence and machine learning communities,
the performance of algorithms for MAP inference on DPPs is often evaluated in terms of the \emph{logarithm} of the determinant \citep{gillenwater2012near,han2020map}.
So, we examine the approximability of the \emph{log-determinant maximization} problem, i.e., $\max_{S \subseteq [n]} \log \det(\mat{A}_S)$.
The unconstrained case is known to admit a $2$-factor approximation algorithm \citep{buchbinder2015tight,buchbinder2018deterministic}.
If the minimum eigenvalue of $\mat{A}$ is at least $1$,
the size-constrained case (i.e., $|S|=k$) is a special case of \emph{monotone submodular maximization}, for which
the \textsc{Greedy} algorithm has an approximation factor of $\frac{\rme}{\rme-1} \approx 1.58$ \citep{han2020map,sharma2015greedy}.
On the other hand, explicit inapproximability results have not been known (to the best of our knowledge), which is the second question.
Note that such hardness results for $\log\det(\mat{A}_S)$ and $\det(\mat{A}_S)$ are ``incomparable'' in the sense that
a multiplicative approximation to one does \emph{not} imply that to the other.

\paragraph{Exponentiated DPPs.}
Given an $n \times n$ positive semi-definite matrix $\mat{A}$,
an \emph{exponentiated DPP (E-DPP)} of exponent $p > 0$
defines a distribution whose probability mass for
$S \subseteq [n]$ is proportional to $\det(\mat{A}_S)^p$ \citep{mariet2018exponentiated}.
We can sharpen or weaken the diversity preference by tuning the value of exponent parameter $p$:
increasing $p$ prefers more diverse subsets than DPPs,,
setting $p=0$ results in a uniform distribution, and E-DPPs coincide with DPPs if $p=1$.
Though computing the normalizing constant, i.e., $\sum_{S \subseteq [n]} \det(\mat{A}_S)^p$, lies at the core of efficient probabilistic inference on E-DPPs,
it seems not to have a closed-form expression.
Currently,
some hardness results on \emph{exact} computation are known
if $p$ is an even integer
\citep{gurvits2005complexity,ohsaka2020intractability}, and
the case of $p \leq 1$ admits a fully polynomial-time randomized approximation scheme (FPRAS)\footnote{An FPRAS is a randomized algorithm 
that outputs an $\rme^{\epsilon}$-approximation with probability at least $\frac{3}{4}$ and runs in polynomial time in the input size and $\epsilon^{-1}$.
} based on an approximate sampler
\citep{anari2019log,robinson2019flexible}.
The third question in this study is to find the value of $p$ such that
the normalizing constant is hard to approximate.

Our research questions in the present study can be summarized as follows:
\begin{framed}
\noindent
\begin{itemize}
\item[\textbf{Q1.}] \emph{Is unconstrained MAP inference on DPPs exponentially inapproximable?}
\item[\textbf{Q2.}] \emph{Is log-determinant maximization constant-factor inapproximable?}
\item[\textbf{Q3.}] \emph{For what value of $p$ is the normalizing constant for E-DPPs inapproximable?}
\end{itemize}
\end{framed}

\input{tab-subdetmax}

\subsection{Our Contributions}
We answer the above questions affirmatively by presenting three complexity-theoretic hardness results.

\paragraph{{(\cref{sec:subdetmax}) Exponential Inapproximability of Unconstrained MAP Inference on DPPs.}}
Our first result is the following (cf.~\cref{tab:subdetmax}):

\begin{framed}
\noindent\textbf{\cref{thm:subdetmax-inapprox}}
(informal)\textbf{.}
\textit{Unconstrained MAP inference on DPPs for an $n \times n$ matrix is \NP-hard to approximate
within a factor of $2^{\beta n}$, where $\beta = 10^{-10^{13}} $.}
\end{framed}

This result significantly improves upon the best-known lower bound of \citet{kulesza2012determinantal}.
Though the universal constant $\beta = 10^{-10^{13}}$ is extremely small,
\cref{thm:subdetmax-inapprox} justifies why any polynomial-factor approximation algorithm for unconstrained MAP inference has not been found.
Our lower bound $2^{\beta n}$ matches the best upper bound $\rme^n$ \citep{nikolov2015randomized}, \emph{up to} a constant in the exponent.
The proof is obtained by carefully extending the proof technique of 
\citet*{civril2013exponential} to the unconstrained case.

\paragraph*{{(\cref{sec:log}) Constant-Factor Inapproximability of Log-Determinant Maximization.}}
We then obtain the following answer to our second question:

\begin{framed}
\noindent
\textbf{\cref{thm:log:inapprox,thm:log:size-inapprox}}
(informal)\textbf{.}
\textit{
Log-determinant maximization is \NP-hard to approximate within a constant factor. In particular,
\begin{itemize}
\item the unconstrained case is \NP-hard to approximate within a factor of $\frac{5}{4}$, and
\item the size-constrained case is \NP-hard to approximate within a factor of $1+10^{-10^{13}}$,
even when the minimum eigenvalue of an input matrix is at least $1$.
\end{itemize}
}
\end{framed}

The above results imply that the existing algorithms' approximation factors are tight, \emph{up to} a constant.
In particular, log-determinant maximization does not admit a polynomial-time approximation scheme (PTAS),\footnote{A PTAS is an approximation algorithm that takes a precision parameter $\epsilon>0$ and outputs
an $\rme^{\epsilon}$-approximation in polynomial time in the input size when $\epsilon$ is fixed.}
unless \cP~$=$~\NP.

\paragraph*{{(\cref{sec:e-dpp}) Exponential Inapproximability of Exponentiated DPPs.}}
Our third result is the following, which is derived by applying \cref{thm:subdetmax-inapprox} (cf.~\cref{fig:edpp}):

\begin{framed}
\noindent
\textbf{\cref{cor:e-dpp}}
(informal)\textbf{.}
\textit{
For every fixed number $p \geq \beta^{-1} = 10^{10^{13}}$,
it is \NP-hard to approximate the normalizing constant for E-DPPs of exponent $p$
for an $n \times n$ matrix within a factor of $2^{\beta pn}$.
Moreover, we cannot generate a sample from E-DPPs of exponent $p$.
}
\end{framed}

This is the first inapproximability result regarding E-DPPs of constant exponent $p$ and
gives a new negative answer to open questions posed by \citet[Section 7.2]{kulesza2012determinantal} and \citet[Section 6]{ohsaka2020intractability}.
The factor $2^{\beta p n}$ is tight up to a constant in the exponent because
$2^{\bigO(pn)}$-factor approximation is possible in polynomial time (\cref{obs:edpp-approx}).
The latter statement means that
in contrast to the case of $p \leq 1$, an efficient approximate sampler does not exist whenever $p \geq 10^{10^{13}}$. 
We stress that when applying
a $(\frac{9}{8}-\epsilon)$-factor inapproximability of \citet*{kulesza2012determinantal} instead of \cref{thm:subdetmax-inapprox},
we would be able to derive inapproximability only if $p = \Omega(n)$ (see \cref{remark:edpp}).

This article is an extended version of our conference paper presented at
the 24th International Conference on Artificial Intelligence and Statistics \citep{ohsaka2021unconstrained},
which includes the following new results on log-determinant maximization in \cref{sec:log}:
\begin{itemize}
\item We improve an inapproximability factor for the unconstrained case from $\frac{8}{7}\approx 1.143$ \citep[Remark 2.10]{ohsaka2021unconstrained} to $\frac{5}{4} = 1.25$ (\cref{thm:log:inapprox}).
\item We prove a constant-factor inapproximability for the size-constrained monotone case (\cref{thm:log:size-inapprox}).
\end{itemize}

\begin{figure}
    \centering
    \includegraphics[trim={0 12.1cm 0 1.3cm},clip,width=\textwidth]{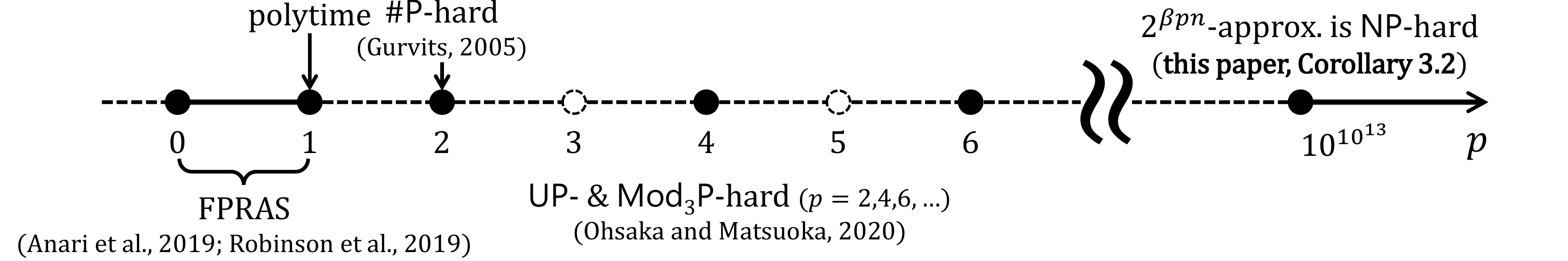}
    \caption{Computational complexity of the normalizing constant for exponentiated DPPs. Our result is $2^{\beta pn}$-factor inapproximability for $p \geq 10^{10^{13}}$.
    Tractability for $p$'s on dashed lines and circles remains open.}
    \label{fig:edpp}
\end{figure}

\subsection{Related Work}\label{subsec:related}

\paragraph{MAP Inference on DPPs.}
In the theoretical computer science community, unconstrained MAP inference on DPPs is known as 
\emph{determinant maximization} (\subdetmax for short).
The size-constrained version ($k$-\subdetmax for short), which restricts the output to size $k$ for parameter $k$, finds applications in computational geometry and discrepancy theory  \citep{nikolov2015randomized}.

On the inapproximability side,
\citet*{ko1995exact} prove that \subdetmax and $k$-\subdetmax are both \NP-hard, and
\citet*{koutis2006parameterized} shows \NP-hardness of approximating the largest $k$-simplex in a $\mathcal{V}$-polytope within a factor of $2^{ck}$ for some $c > 0$, implying that
$k$-\subdetmax is also exponentially inapproximable in $k$.
\citet*{civril2013exponential} directly prove a similar result for $k$-\subdetmax.
\citet*{summa2014largest} study the special case of $k$-\subdetmax, where $k$ is fixed to the rank of an input matrix,
which is still \NP-hard to approximate within $(2^{\frac{1}{506}}-\epsilon)^{k}$ for any $\epsilon > 0$.
\citet*{kulesza2012determinantal} use
the reduction technique developed by \citet*{civril2009selecting} to show an inapproximability factor of ($\frac{9}{8}-\epsilon$) for \subdetmax for any $\epsilon > 0$,
which is the current best lower bound.
    
On the algorithmic side,
\citet*{civril2009selecting} prove that the \textsc{Greedy} algorithm for $k$-\subdetmax achieves an approximation factor of $k!^2 = 2^{\bigO(k \log k)}$.
In their celebrated work,
\citet*{nikolov2015randomized} gives an $\rme^{k}$-approximation algorithm for $k$-\subdetmax; this is the current best approximation factor.
Invoking the algorithm of \citet{nikolov2015randomized} for all $k$ immediately yields an $\rme^{n}$-approximation for \subdetmax.
The recent work focuses on \subdetmax under complex constraints such as partition constraints \citep{nikolov2016maximizing} and matroid constraints \citep{madan2020maximizing}.

In artificial intelligence and machine learning applications,
unconstrained MAP inference is preferable if we do not (or cannot) prespecify the desired size of output, e.g.,
tweet timeline generation \citep{yao2016tweet},
object detection \citep{lee2016individualness},
change-point detection \citep{zhang2016block}, and
others \citep{gillenwater2012near,han2017faster,chen2018fast,chao2015large}.
Since the log-determinant as a set function $f(S) \triangleq \log \det(\mat{A}_S)$ for a positive semi-definite matrix $\mat{A}$ is submodular,\footnote{We say that a set function $f: 2^{[n]} \to \bbR$ is \emph{submodular} if $f(S) + f(T) \geq f(S \cup T) + f(S \cap T)$ for all $S,T \subseteq [n]$.}
the \textsc{Greedy} algorithm for monotone submodular maximization \citep{nemhauser1978analysis} is widely used,
which guarantees an $(\frac{\rme}{\rme-1})$-factor approximation (with respect to $f$) under a size constraint if 
the minimum eigenvalue of $\mat{A}$ is at least $1$; i.e.,
$f$ is a nonnegative, monotone, and submodular function \citep{han2020map,sharma2015greedy}.
Though the log-determinant $f$ is \emph{not} necessarily monotone, for which \textsc{Greedy} has, in fact, a poor approximation guarantee,
it works pretty well in practice \citep{yao2016tweet,zhang2016block}.
Several attempts have been made to scale up \textsc{Greedy} \citep{han2017faster,chen2018fast,han2020map,gartrell2020scalable},
whose naive implementation requires quartic time in $n$.
Other than \textsc{Greedy},
\citet*{gillenwater2012near} propose a gradient-based efficient algorithm with an approximation factor of $4$ (with respect to $f$).
Subsequently, 
a $2$-factor approximation algorithm for unconstrained nonmonotone submodular maximization
(including $\max_{S \subseteq [n]} f(S)$ as a special case) is developed \citep{buchbinder2015tight,buchbinder2018deterministic}.
Our results give constant-factor inapproximability for log-determinant maximization.

\paragraph{Exponentiated DPPs.}
We review known results on the computational complexity of the normalizing constant of E-DPPs, i.e.,
$\sum_{S \subseteq [n]} \det(\mat{A}_S)^p$ for $\mat{A} \in \bbQ^{n \times n}$,
briefly appearing in \citet{zou2012priors}; \citet{gillenwater2014approximate}.
The case of $p=1$ enjoys a simple closed-form expression that $\sum_{S \subseteq [n]}\det(\mat{A}_S) = \det(\mat{A}+\mat{I})$ \citep{kulesza2012determinantal}.
Such a closed-form is unknown if $p < 1$, but a Markov chain Monte Carlo algorithm mixes in polynomially many steps
thanks to their log-concavity \citep{anari2019log,robinson2019flexible}, implying an FPRAS.

On the other hand,
the case of $p > 1$ seems a little more difficult.
\citet*{kulesza2012determinantal} posed efficient computation of the normalizing constant for E-DPPs as an open question.
Surprisingly, \citet*{gurvits2005complexity,gurvits2009complexity} has proven that
computing $\sum_{S \subseteq [n]} \det(\mat{A}_S)^2$ for a P-matrix $\mat{A}$ is \shP-hard, but
it is approximable within an $\rme^n$-factor \citep*{anari2017generalization}.
\citet*{mariet2018exponentiated} derive an upper bound on the mixing time of sampling algorithms parameterized by $p$ and eigenvalues of $\mat{A}$.
\citet*{ohsaka2020intractability} derive \UP-hardness and \ModkP{3}-hardness for every positive even integer $p=2,4,6,\ldots$.
On the positive side,
\citet{ohsaka2020intractability} develop $r^{\bigO(pr)} n^{\bigO(1)}$-time algorithms for integer exponent $p$, where
$r$ is the rank or the treewidth of $\mat{A}$.
Our study strengthens previous work by giving the first inapproximability result for every (fixed) constant exponent $p \geq 10^{10^{13}}$.

\subsection{Notations and Definitions}
For a positive integer $n$, let $[n] \triangleq \{1,2,\ldots, n\}$.
For a finite set $S$ and an integer $k$, we write ${S \choose k}$ for the family of all size-$k$ subsets of $S$.
The Euclidean norm is denoted $\| \cdot \|$;
i.e., $\|\vec{v}\| \triangleq \sqrt{\sum_{i \in [d]}(v(i))^2} $ for a vector $\vec{v}$ in $\bbR^{d}$.
We use $\langle \cdot, \cdot \rangle$
for the standard inner product; i.e.,
$ \langle \vec{v}, \vec{w} \rangle \triangleq \sum_{i \in [d]} v(i) \cdot w(i) $ for two vectors $\vec{v}, \vec{w}$ in $\bbR^{d}$.
For an $n \times n$ matrix $\mat{A}$ and an index set $S \subseteq [n]$,
we use $\mat{A}_S$ to denote the principal submatrix of $\mat{A}$ whose
rows and columns are indexed by $S$.
For a matrix $\mat{A}$ in $\bbR^{n \times n}$,
its \emph{determinant} is defined as
\begin{align*}
    \det(\mat{A}) \triangleq \sum_{\sigma \in \mathfrak{S}_n} \sgn(\sigma) \prod_{i \in [n]} A_{i,\sigma(i)},
\end{align*}
where $\mathfrak{S}_n$ is the symmetric group on $[n]$, and
$\sgn(\sigma)$ is the sign of a permutation $\sigma$.
We define $\det(\mat{A}_{\emptyset}) \triangleq 1$.
For a set $\mat{V} = \{\vec{v}_1, \ldots, \vec{v}_n\}$ of $n$ vectors in $\bbR^{d}$,
the \emph{volume} of the parallelepiped spanned by $\mat{V}$ is defined as
\begin{align}\label{eq:def-vol}
    \vol(\mat{V}) \triangleq \|\vec{v}_1\| \cdot
    \prod_{2 \leq i \leq n} \left\| \vec{v}_i - \proj_{\{\vec{v}_1, \ldots, \vec{v}_{i-1}\}} (\vec{v}_i) \right\|,
\end{align}
where $\proj_{\mat{P}}(\cdot)$ is an operator of orthogonal projection onto the subspace spanned by vectors in $\mat{P}$.
We finally introduce the \emph{AM--GM inequality}: for any $n$ nonnegative real numbers $x_1, \ldots, x_n$, we have
\begin{align*}
    \frac{1}{n}\sum_{i \in [n]}x_i \geq \Biggl(\prod_{i \in [n]}x_i \Biggr)^{\frac{1}{n}}.
\end{align*}

\section{Exponential Inapproximability of Unconstrained MAP Inference}
\label{sec:subdetmax}

We prove an exponential-factor inapproximability result for 
unconstrained MAP inference on DPPs, which is identical to the following
determinant maximization problem:
\begin{definition}
Given a positive semi-definite matrix $\mat{A}$ in $\bbQ^{n \times n}$,
\emph{determinant maximization} (\subdetmax) asks to find a subset $S$ of $[n]$ such that
the determinant $\det(\mat{A}_S)$ of a principal submatrix is maximized.
The optimal value of \subdetmax is denoted $\maxdet(\mat{A}) \triangleq \max_{S \subseteq [n]} \det(\mat{A}_S)$.
\end{definition}
    
We say that a polynomial-time algorithm $\ALG$ is
a \emph{$\rho$-approximation algorithm} for $\rho \geq 1$ if for all $\mat{A} \in \bbQ^{n \times n}$, it holds that
\begin{align*}
    \det(\ALG(\mat{A})) \geq \left(\frac{1}{\rho}\right) \cdot \maxdet(\mat{A}),
\end{align*}
where $\ALG(\mat{A})$ is the output of $\ALG$ on $\mat{A}$.
The factor $\rho$ can be a function in the input size $n$, e.g., $\rho(n) = 2^{n}$, and
(asymptotically) smaller $\rho$ is a better approximation factor.
We also define \gapsubdetmax{$[s(n),c(n)]$} for two functions $c(n)$ and $s(n)$ as a problem of deciding
whether $\maxdet(\mat{A}) \geq c(n)$ or $\maxdet(\mat{A}) < s(n)$.\footnote{Precisely, $\mat{A}$ is ``promised'' to satisfy either
$\maxdet(\mat{A}) \geq c(n)$ (\emph{yes} instance) or
$\maxdet(\mat{A}) < s(n)$ (\emph{no} instance).
Such a problem is called a \emph{promise} problem.
}
If \gapsubdetmax{$[s(n),c(n)]$} is \NP-hard, then
so is approximating \subdetmax within a factor of $\frac{c(n)}{s(n)}$.

We are now ready to state our result formally.
\begin{theorem}\label{thm:subdetmax-inapprox}
There exist universal constants $\ccc$ and $\sss$ such that 
$\ccc - \sss > 10^{-10^{13}}$ and \gapsubdetmax{$[2^{\sss n}, 2^{\ccc n}]$} is \NP-hard,
where $n$ is the order of an input matrix.
In particular, it is \NP-hard to approximate \subdetmax within a factor of $2^{\beta n}$,
where $\beta = 10^{-10^{13}}$.
\end{theorem}
\begin{remark}
The universal constant $\beta = 10^{-10^{13}}$ is so extremely small that $2^{\beta n} \approx 1$ for real-world matrices, whose possible size $n$ is limited inherently.
The significance of \cref{thm:subdetmax-inapprox} is that
it can rule out the existence of any polynomial-factor approximation algorithm (unless \cP~$=$~\NP).
As a corollary, we also show inapproximability for E-DPPs of constant exponent $p$.
\end{remark}
    
The input for \subdetmax is often given by a \emph{Gram matrix} $\mat{A} \in \bbQ^{n \times n}$,
where $ A_{i,j} = \langle \vec{v}_i, \vec{v}_j \rangle $ for all $i,j \in [n]$ for
$n$ vectors $\vec{v}_1, \ldots, \vec{v}_n$ in $\bbQ^{d}$.
In such a case, we have a simple relation between the principal minor and the volume of the parallelepiped that
$\det(\mat{A}_S) = \vol(\{\vec{v}_i\}_{i \in S})^2$ for every subset $S \subseteq [n]$.
\subdetmax is thus essentially equivalent to the following optimization problem:

\begin{definition}
Given a set $\mat{V} = \{\vec{v}_1, \ldots, \vec{v}_n\}$ of $n$ vectors in $\bbQ^{d}$,
\emph{volume maximization} (\volmax) asks to find a subset $\mat{S}$ of $\mat{V}$ such that
the volume $\vol(\mat{S})$ is maximized.
The optimal value of \volmax is denoted
$\maxvol(\mat{V}) \triangleq \max_{S \subseteq [n]} \vol(\{\vec{v}_i\}_{i \in S})$.
\end{definition}
Observe that
there exists a $\rho(n)$-approximation algorithm for \subdetmax \emph{if and only if}
there exists a $\sqrt{\rho(n)}$-approximation algorithm for \volmax.
    
\paragraph{Outline of the Remainder of \cref{sec:subdetmax}.}
\cref{subsec:game-rep} introduces 
projection games to be reduced to \volmax and
\citeauthor{raz1998parallel}'s parallel repetition theorem.
\cref{subsec:inapprox-game} reviews the indistinguishability of projection games.
\cref{subsec:main-lemma} describes \nameref{lem:main},
which is crucial in proving \cref{thm:subdetmax-inapprox}, and
\cref{subsec:proof-lemma} is devoted to the proof of \nameref{lem:main}.

\subsection{Projection Game and  Parallel Repetition Theorem}\label{subsec:game-rep}
We introduce projection games followed by the parallel repetition theorem.
    
\begin{definition}
A \emph{2-player 1-round projection game} is specified by a quintuple
$ \GG = (X,Y,E,\Sigma,\Pi) $ such that the following conditions are satisfied:
\begin{itemize}
    \item $(X,Y,E)$ is a bipartite graph with vertex sets $X$ and $Y$ and an edge set $E$ between $X$ and $Y$,
    \item $\Sigma$ is an alphabet, and
    \item $\Pi = \{\pi_e\}_{e \in E} $ is a constraint set, 
    where $\pi_e$ for each edge $e \in E$ is a function $\Sigma \to \Sigma$.
\end{itemize}
\end{definition}

A \emph{labeling} $\sigma$ is defined as a label assignment of each vertex of the bipartite graph,
i.e., $\sigma: (X \uplus Y) \to \Sigma$.
An edge $e = (x,y) \in E$ is said to be \emph{satisfied} by $\sigma$ if $\pi_e(\sigma(x)) = \sigma(y) $.
The \emph{value} of a projection game $\GG$, denoted $\val(\GG)$, is defined as
the maximum fraction of edges satisfied over all possible labelings $\sigma$, i.e.,
\begin{align*}
    \val(\GG) \triangleq \max_{\sigma: (X \uplus Y) \to \Sigma} \frac{1}{|E|} \sum_{e=(x,y) \in E} \Bigl[\!\!\Bigl[ \pi_{e}(\sigma(x)) = \sigma(y) \Bigr]\!\!\Bigr].
\end{align*}
The \prb{LabelCover} problem is defined as
finding a labeling that satisfies the maximum fraction of edges in the bipartite graph of a projection game.

We then define the \emph{product} of two games followed by the parallel repetition of a game.
\begin{definition}
Let $\GG_1 = (X_1,Y_1,E_1,\Sigma_1,\{\pi_{1,e}\}_{e \in E_1})$ and
$\GG_2 = (X_2,Y_2,E_2,\Sigma_2,\{\pi_{2,e}\}_{e \in E_2})$
be two projection games.
The \emph{product} of $\GG_1$ and $\GG_2$, denoted $ \GG_1 \otimes \GG_2 $, is defined as a new game
$ (X_1 \times X_2, Y_1 \times Y_2, E, \Sigma_1 \times \Sigma_2, \Pi = \{\pi_e\}_{e \in E}) $, where
$E \triangleq \{ ( (x_1,x_2), (y_1,y_2) ) \mid (x_1,y_1) \in E_1, (x_2,y_2) \in E_2 \}$, and
for each edge $e = ((x_1,x_2), (y_1,y_2)) \in E$,
$\pi_e : \Sigma_1 \times \Sigma_2 \to \Sigma_1 \times \Sigma_2$ is defined as
$\pi_e((i_1,i_2)) \triangleq (\pi_{1,(x_1,y_1)}(i_1), \pi_{2,(x_2,y_2)}(i_2)) $
for labels $i_1 \in \Sigma_1$ and $i_2 \in \Sigma_2$.
\end{definition}
    
The \emph{$\ell$-hold parallel repetition} of $\GG$ for any positive integer $\ell$
is defined as
\begin{align*}
\GG^{\otimes \ell} \triangleq \underbrace{\GG \otimes \cdots \otimes \GG}_{\ell \text{ times}}.
\end{align*}
\citet*{raz1998parallel} proved the \emph{parallel repetition theorem}, which states that
for every (not necessarily projection) game $\GG$ with $\val(\GG)=1-\epsilon$, it holds that
$\val(\GG^{\otimes \ell}) \leq (1-\overline{\epsilon})^{\frac{\ell}{\log |\Sigma|}}$,
where $\overline{\epsilon}$ is a constant depending only on $\epsilon$.
Once establishing an inapproximability factor for \prb{LabelCover}, we can amplify it by taking parallel repetition.
The use of the parallel repetition theorem has led to
inapproximability results for many (\NP-hard) optimization problems,
such as \prb{SetCover} \citep{feige1998threshold} and
\prb{MaxClique} \citep{hastad1999clique}.
We refer to a tighter, explicit bound derived by \citet*{dinur2014analytical}.

\begin{theorem}[\protect{\citealp[Corollary 1]{dinur2014analytical}}]
\label{thm:dinur-steurer}
For any projection game $\GG$ with $\val(\GG) \leq 1-\epsilon$ for some $\epsilon > 0$,
the $\ell$-fold parallel repetition of $\GG$ satisfies that
\begin{align*}
    \val(\GG^{\otimes \ell}) \leq \left(1-\frac{\epsilon^2}{16}\right)^{\ell}.
\end{align*}
\end{theorem}

\subsection{Indistinguishability of Projection Games}
\label{subsec:inapprox-game}
We review the indistinguishability of (the value of) projection games,
in other words, inapproximability of \prb{LabelCover}.
The following theorem shows that
we cannot decide in polynomial time whether a projection game has a value $1$ or has a value less than $1-\epsilon$ for some $\epsilon > 0$.
Though its proof is widely known, we include it in \cref{app:inapprox-labelcover}
to explicitly describe the value of such $\epsilon$.

\begin{theorem}[See, e.g., \citealp{feige1998threshold,hastad2001some,trevisan2004inapproximability,vazirani2013approximation,tamaki2015parallel}]
\label{thm:inapprox-labelcover}
Let $\GG = (X,Y,E,\Sigma,\Pi)$ be a projection game such that
$(X,Y,E)$ is a $15$-regular bipartite graph
(i.e., each vertex of $X \uplus Y$ is incident to exactly $15$ edges),
where $|X| = |Y| = 5n$ and $|E|=75n$ for some positive integer $n$ divisible by $3$, and $|\Sigma| = 7$.
Then, it is \NP-hard to distinguish between
$\val(\GG) = 1$ and $\val(\GG) < 1-\frac{1}{206{,}401}$.
\end{theorem}
A projection game satisfying the conditions in \cref{thm:inapprox-labelcover}
is said to be \emph{special} in this paper.
Owing to \cref{thm:inapprox-labelcover,thm:dinur-steurer}, for any $\ell$,
it is \NP-hard to decide whether
the $\ell$-fold parallel repetition $\GG^{\otimes \ell}$ of a special projection game satisfies
$\val(\GG^{\otimes \ell}) = 1$ or
\begin{align*}
    \val(\GG^{\otimes \ell}) <
    \left(1-\frac{1}{(206{,}401)^2 \cdot 16}\right)^{\ell} < 2^{-2 \cdot 10^{-12} \ell}.
\end{align*}
Hereafter, we let $\alpha \triangleq 2 \cdot 10^{-12}$.

\subsection{\nameref{lem:main} and Proof of \cref{thm:subdetmax-inapprox}}
\label{subsec:main-lemma}

The proof of \cref{thm:subdetmax-inapprox} relies on
a reduction from the $\ell$-fold parallel repetition of a special projection game to \volmax.
Throughout the remainder of this section,
we fix the value of $\ell$ as
$\displaystyle \ell \triangleq \left\lceil \frac{4}{\alpha} \right\rceil = 2 \cdot 10^{12}.$
Our main lemma in the following can be thought of as an extension of \citet*{civril2013exponential}.
The proof is deferred to the next subsection.

\begin{lemma}[Main Lemma]
\label{lem:main}
There is a polynomial-time reduction from
the $\ell$-fold parallel repetition $\GG^{\otimes \ell}$ of a special projection game to
an instance $\mat{V} =\{\vec{v}_1, \ldots, \vec{v}_N\}$ of \volmax such that
$N = 2 \cdot (35n)^{\ell}$ for some integer $n$,
each vector of $\mat{V}$ is normalized (i.e., $\|\vec{v}_i\| = 1$ for all $i \in [N]$), and
the following conditions are satisfied:
\begin{itemize}
    \item (Completeness) If $\val(\GG^{\otimes \ell})=1$,
    then there exists a set $\mat{S}$ of $K$ vectors from $\mat{V}$ with volume
    $\vol(\mat{S}) = 1$, where $K = \frac{N}{7^{\ell}}$.
    \item (Soundness) If $\val(\GG^{\otimes \ell}) < 2^{-\alpha \ell}$,
    then any set $\mat{S}$ of $k$ vectors from $\mat{V}$ satisfies
    the following properties:
    \begin{itemize}
        \item[\textbf{1.}]~$0 \leq k < \frac{4}{5}K$: $\vol(\mat{S}) \leq 1$.
        \item[\textbf{2.}]~$\frac{4}{5}K \leq k \leq N$:
        $\vol(\mat{S}) < 2^{-\beta^{\circ} k}$, where $\beta^{\circ} = 10^{-10^{12.4}}$.
    \end{itemize}
\end{itemize}
\end{lemma}
By \cref{lem:main}, we can prove \cref{thm:subdetmax-inapprox} as follows.

\begin{proof}[Proof of \cref{thm:subdetmax-inapprox}]
Let $\mat{V} = \{\vec{v}_1, \ldots, \vec{v}_N\}$
be an instance of \volmax reduced from
the $\ell$-fold parallel repetition $\GG^{\otimes \ell}$ by \cref{lem:main}.
Create a new instance of \volmax
$\mat{W} = \{\vec{w}_1, \ldots, \vec{w}_N\}$,
where $\vec{w}_i = 2^{\beta^{\circ}} \cdot \vec{v}_i$ for each $i \in [N]$
(which is a polynomial-time reduction as $2^{\beta^{\circ}}$ is constant).
If $\val(\GG^{\otimes \ell}) = 1$, then there is a set
$\mat{S}$ of $K= N/7^{\ell}$ vectors from $\mat{W}$ such that
$\vol(\mat{S}) = 2^{\beta^{\circ} K}$.
On the other hand, if $\val(\GG^{\otimes \ell}) < 2^{-\alpha \ell}$,
then $\maxvol(\mat{W})$ is (strictly) bounded from above by $ 2^{\frac{4}{5} \beta^{\circ} K} $
through the following case analysis on the size of $\mat{S} \subseteq \mat{W}$:
\begin{itemize}
    \item \textbf{1.}~$0 \leq |\mat{S}| < \frac{4}{5}K$:
    $ \vol(\mat{S}) < 2^{\frac{4}{5}\beta^{\circ} K} $.
    \item \textbf{2.}~$\frac{4}{5} K \leq |\mat{S}| \leq N$:
    $\vol(\mat{S}) < 2^{(\beta^{\circ}-\beta^{\circ})|\mat{S}|} \leq 1 $.
\end{itemize}
\item
It is thus \NP-hard to decide whether
$\maxvol(\mat{W}) \geq 2^{\frac{\beta^{\circ}}{7^{\ell}} N}$ or
$\maxvol(\mat{W}) < 2^{\frac{4 \beta^{\circ}}{5 \cdot 7^{\ell}} N}$
by \cref{thm:dinur-steurer,thm:inapprox-labelcover}.
Owing to the relation between \volmax and \subdetmax,
\gapsubdetmax{$[2^{\frac{8 \beta^{\circ}}{5 \cdot 7^{\ell}}n},2^{\frac{2\beta^{\circ}}{7^{\ell}}n}]$} is also \NP-hard, where $n$ is the order of an input matrix.
In particular, it is \NP-hard to approximate \subdetmax within a factor of $2^{(\ccc-\sss)n}$, where
$\ccc = \frac{2\beta^{\circ}}{7^{\ell}}$ and
$\sss = \frac{8\beta^{\circ}}{5 \cdot 7^{\ell}}$.
Observing that
$\beta = 10^{-10^{13}} < 10^{-10^{12.7}} < \ccc - \sss$ suffices to complete the proof.
\end{proof}

\subsection{Proof of \nameref{lem:main}}
\label{subsec:proof-lemma}
We now prove \nameref{lem:main}.
We first introduce tools from \citet*{civril2013exponential}.

\begin{lemma}[\protect{\citealp[Union Lemma]{civril2013exponential}}]
\label{lem:civril-union}
Let $\mat{P}$ and $\mat{Q}$ be two (finite) sets of vectors in $\bbR^{d}$.
Then, we have the following:
\begin{align*}
    \vol(\mat{P} \cup \mat{Q}) \leq \vol(\mat{Q}) \cdot \prod_{\vec{v} \in \mat{P}} \dis(\vec{v}, \mat{Q}),
\end{align*}
where $\dis(\vec{v}, \mat{Q})$ denotes the distance of $\vec{v}$ to the subspace spanned by $\mat{Q}$; i.e.,
\begin{align*}
    \dis(\vec{v}, \mat{Q}) \triangleq \left\| \vec{v} - \proj_{\mat{Q}}(\vec{v}) \right\|.
\end{align*}
\end{lemma}

\begin{lemma}[\protect{\citealp[Lemma~13]{civril2013exponential}}]
\label{lem:civril-vector}
For any positive integer $\ell$, there exists a set of $2^{\ell}$ vectors
$\mat{B}^{(\ell)} = \{\vec{b}_1, \ldots, \vec{b}_{2^{\ell}}\}$
of dimension $2^{\ell+1}$ such that the following conditions are satisfied:
\begin{itemize}
    \item Each element of vectors is either $0$ or $2^{-\frac{\ell}{2}}$.
    \item $\|\vec{b}_i\| = 1 $ for all $i \in [2^{\ell}]$.
    \item $\langle \vec{b}_i, \vec{b}_j \rangle = \frac{1}{2}$ for all $i, j\in [2^\ell]$ with $i \neq j$.
    \item $ \langle \vec{b}_i, \overline{\vec{b}_j} \rangle = \frac{1}{2} $ for all $i,j \in [2^\ell]$ with $i \neq j$, where
    $ \overline{\vec{b}_j} \triangleq 2^{-\frac{\ell}{2}} \cdot \vec{1} - \vec{b}_j $.
    Note that $\langle \vec{b}_i, \overline{\vec{b}_i} \rangle = 0$.
\end{itemize}
Moreover, $\mat{B}^{(\ell)}$ can be constructed in time $\bigO(4^{\ell})$.
\end{lemma}

\paragraph{Our Reduction.}\label{par:reduction}
We explain how to reduce from special projection games to \volmax.
Let $\GG^{\otimes \ell} = (X,Y,E,\Sigma,\Pi)$ be the $\ell$-fold parallel repetition of a special projection game.
By definition,
$(X,Y,E)$ is a $15^{\ell}$-regular bipartite graph, where
$|X|=|Y|=(5n)^{\ell}$,
$|E|=(75n)^{\ell}$, and
$|\Sigma|=7^{\ell}$ for some integer $n$.
Assume that $\Sigma = [7^{\ell}]$ for notational convenience.

For each pair of a vertex of $X \uplus Y$ and a label of $\Sigma$, we define a vector as follows.
Each vector consists of $|E|$ blocks,
each of which is $2^{3\ell+1}$-dimensional and is either a vector in the set
$\mat{B}^{(3\ell)} = \{ \vec{b}_1, \ldots, \vec{b}_{2^{3\ell}} \}$ or the zero vector $\vec{0}$.
Let $\vec{v}_{x,i}$ (resp.~$\vec{v}_{y,i}$) denote the vector for a pair $(x, i) \in X \times \Sigma$
(resp.~$(y,i) \in Y \times \Sigma$), and
let $ \vec{v}_{x,i}(e)$ (resp.~$\vec{v}_{y,i}(e)$) denote the block of $\vec{v}_{x,i}$ (resp.~$\vec{v}_{y,i}$) corresponding to edge $e \in E$.
Each block is defined as follows:
\begin{align*}
    \vec{v}_{x,i}(e) & \triangleq
    \begin{cases}
    \displaystyle\frac{\overline{\vec{b}_{\pi_e(i)}}}{15^{\ell/2}} & \text{if } e \text{ is incident to } x,  \\
    \vec{0} & \text{otherwise},
    \end{cases} \\
    \vec{v}_{y,i}(e) & \triangleq
    \begin{cases}
    \displaystyle\frac{\vec{b}_i}{15^{\ell/2}} & \text{if } e \text{ is incident to } y, \\
    \vec{0} & \text{otherwise}.
    \end{cases}
\end{align*}
Since each vector contains exactly $15^{\ell}$ blocks chosen from $\mat{B}^{(3\ell)}$, it is normalized;
i.e., $\|\vec{v}_{x,i}\| = \|\vec{v}_{y,i}\| = 1$ for all $x \in X, y \in Y, i \in \Sigma$.
Note that
$\vec{v}_{x_1, i_1}$ and $\vec{v}_{x_2,i_2}$ (resp.~$\vec{v}_{y_1, i_1}$ and $\vec{v}_{y_2,i_2}$)
are orthogonal for any $x_1,x_2 \in X$ (resp.~$y_1,y_2 \in Y$) and
$i_1,i_2 \in \Sigma$ if $x_1 \neq x_2$ (resp.~$y_1 \neq y_2$), and
$\vec{v}_{x,i}$ and $\vec{v}_{y,j}$ for $x \in X, y \in Y, i,j \in \Sigma$ are orthogonal if
$(x,y) \in E$ and $\pi_{(x,y)}(i) = j$ as
$\langle \vec{v}_{x,i}, \vec{v}_{y,j} \rangle = \frac{1}{15^{\ell}} \langle \overline{\vec{b}_{\pi_{(x,y)}(i)}}, \vec{b}_{j} \rangle = 0$,
or if $ (x,y) \not \in E $.

We then define an instance $\mat{V}$ of \volmax as follows:
\begin{align*}
    \mat{V} \triangleq
    \{ \vec{v}_{x,i} \mid x \in X, i \in \Sigma \} \uplus
    \{ \vec{v}_{y,i} \mid y \in Y, i \in \Sigma \}.
\end{align*}
Here, $\mat{V}$ contains $N \triangleq 2 \cdot (35n)^{\ell} $ vectors.
Define $K \triangleq |X|+|Y| = 2 \cdot (5n)^{\ell}$;
it holds that $K = N / 7^{\ell}$.
Construction of $\mat{V}$ from $\GG^{\otimes \ell}$ can be done in polynomial time in $n$.
In what follows,
we show that 
$\mat{V}$ satisfies the conditions listed in \nameref{lem:main}.

\paragraph{Completeness.}

\begin{lemma}\label{lem:completeness}
If $\val(\GG^{\otimes \ell}) = 1$, then there exists a set $\mat{S}$ of $K$ vectors from $\mat{V}$ such that $\vol(\mat{S}) = 1$.
\end{lemma}
\begin{proof}
Let $\sigma: (X \uplus Y) \to \Sigma$ be an (optimal) labeling satisfying all the edges of $E$.
For each edge $ e=(x, y) \in E$, we have
$\langle \vec{v}_{x,\sigma(x)}, \vec{v}_{y,\sigma(y)} \rangle = 0$ since $ \pi_e(\sigma(x)) = \sigma(y) $.
We further have
$\langle \vec{v}_{x_1,\sigma(x_1)}, \vec{v}_{x_2,\sigma(x_2)} \rangle = 0$ for $x_1, x_2 \in X$
whenever $x_1 \neq x_2$,
$\langle \vec{v}_{y_1,\sigma(y_1)}, \vec{v}_{y_2,\sigma(y_2)} \rangle = 0$ for $y_1, y_2 \in Y$
whenever $y_1 \neq y_2$, and
$\langle \vec{v}_{x,\sigma(x)}, \vec{v}_{y, \sigma(y)} \rangle = 0$
for $x \in X, y \in Y$
whenever $(x,y) \not \in E$.
Hence, $K$ vectors in the set defined as
$\mat{S} \triangleq \{\vec{v}_{x,\sigma(x)} \mid x \in X \} \uplus \{\vec{v}_{y,\sigma(y)} \mid y \in Y\}$ 
are orthogonal to each other,
implying $\vol(\mat{S}) = 1$.
\end{proof}

\paragraph{Soundness.}
Different from \citet*{civril2013exponential},
we need to bound the volume of \emph{every} subset $\mat{S} \subseteq \mat{V}$.
We consider two cases:
\textbf{1.}~$0 \leq |\mat{S}| < \frac{4}{5} K$ and
\textbf{2.}~$\frac{4}{5}K \leq |\mat{S}| \leq N$.

\paragraph{Soundness 1.~$0 \leq |\mat{S}| < \frac{4}{5}K$.}
\begin{lemma}\label{lem:soundness-1}
Suppose $\val(\GG^{\otimes \ell}) < 2^{-\alpha \ell}$.
For any set $\mat{S}$ of less than $\frac{4}{5}K$ vectors from $\mat{V}$,
it holds that $\vol(\mat{S}) \leq 1$.
\end{lemma}
\begin{proof}
The proof is a direct consequence of the fact that
every vector of $\mat{V}$ is normalized.
\end{proof}

\paragraph{Soundness 2.~$\frac{4}{5}K \leq |\mat{S}| \leq N$.}

For a set $\mat{S}$ of vectors from $\mat{V}$,
we use the following notations:
\begin{align*}
\mat{S}_X & \triangleq \{ \vec{v}_{x,i} \in \mat{S} \mid x \in X, i \in \Sigma \}, 
& \mat{S}_Y & \triangleq \{ \vec{v}_{y,i} \in \mat{S} \mid y \in Y, i \in \Sigma \}, \\
X(\mat{S}) & \triangleq \{ x \in X \mid \exists i \in \Sigma, \vec{v}_{x,i} \in \mat{S} \}, 
& Y(\mat{S}) & \triangleq \{ y \in Y \mid \exists i \in \Sigma, \vec{v}_{y,i} \in \mat{S} \}, \\
\rep(\mat{S}_X) & \triangleq |\mat{S}_X| - |X(\mat{S})|,
& \rep(\mat{S}_Y) & \triangleq |\mat{S}_Y| - |Y(\mat{S})|.
\end{align*}
Here, $\rep(\mat{S}_X)$ and $\rep(\mat{S}_Y)$ mean
how many times the same vertex appears (i.e., the number of repetitions) in
the vectors of $\mat{S}_X$ and $\mat{S}_Y$, respectively.
The following lemma given by \citet{civril2013exponential}
bounds the volume of $\mat{S}_X$ and $\mat{S}_Y$ in terms of
$\rep(\mat{S}_X)$ and $\rep(\mat{S}_Y)$, respectively,
which is proved in \cref{app:proofs} for the sake of completeness.
\begin{lemma}[\protect{\citealp[Lemma 16]{civril2013exponential}}]
\label{lem:civril-vol}
For any set $\mat{S}$ of vectors from $\mat{V}$,
it holds that
\begin{align*}
    \vol(\mat{S}_X) \leq \left(\frac{\sqrt{3}}{2}\right)^{\rep(\mat{S}_X)} \text{ and }
    \vol(\mat{S}_Y) \leq \left(\frac{\sqrt{3}}{2}\right)^{\rep(\mat{S}_Y)}.
\end{align*}
\end{lemma}

We first show that
both $X(\mat{S})$ and $Y(\mat{S})$ contain $\Omega(k)$ vertices if their volume is sufficiently large.

\begin{claim}
\label{clm:S-bound}
For any set $\mat{S} $ of $k \geq \frac{4}{5} K$ vectors from $\mat{V}$,
if $\vol(\mat{S}) \geq 2^{-ck}$ for some number $c > 0$, then it holds that
\begin{align*}
|X(\mat{S})| > \left(\frac{3}{8}-10c\right)k \text{ and }
|Y(\mat{S})| > \left(\frac{3}{8}-10c\right)k.    
\end{align*}
\end{claim}
\begin{proof}
Observe first that $\vol(\mat{S}_X) \geq \vol(\mat{S}) \geq 2^{-ck}$.
By \cref{lem:civril-vol}, we have
$
    (\sqrt{3}/2)^{\rep(\mat{S}_X)}
    \geq \vol(\mat{S}_X) \geq 2^{-ck},
$
implying $ \rep(\mat{S}_X) \leq ck / \log_2(2/\sqrt{3}) < 5ck $. Similarly, we have $\rep(\mat{S}_Y) < 5ck$.
Using the facts that
$|\mat{S}_X| = |X(\mat{S})| + \rep(\mat{S}_X)$,
$|\mat{S}_Y| = |Y(\mat{S})| + \rep(\mat{S}_Y)$, and
$k = |\mat{S}_X| + |\mat{S}_Y|$,
we bound $|X(\mat{S})|$ from below as follows:
\begin{align*}
|X(\mat{S})|
& = k - |\mat{S}_Y| - \rep(\mat{S}_X) = k-|Y(\mat{S})|-\rep(\mat{S}_Y) - \rep(\mat{S}_X) \\
& > k - |Y| - 10ck = \left( 1 - \frac{(5n)^{\ell}}{k} - 10c \right) k \\
& \geq \left(\frac{3}{8} - 10c\right)k,
\end{align*}
where the last inequality follows from the fact that $k \geq \frac{4}{5}K$ and $K = 2 \cdot (5n)^{\ell}$.
Similarly, we have
$|Y(\mat{S})| > (\frac{3}{8} - 10c)k$.
\end{proof}

We now show that
no vector set has a volume close to $1$ if $\val(\GG^{\otimes \ell})$ is small.
\begin{lemma}\label{lem:soundness-2}
Suppose $\val(\GG^{\otimes \ell}) < 2^{-\alpha \ell}$.
For any set $\mat{S}$ of $k$ vectors from $\mat{V}$ with $k \geq \frac{4}{5}K$,
it holds that $\vol(\mat{S}) < 2^{-\beta^{\circ} k}$,
where $\beta^{\circ} = 10^{-10^{12.4}}$.
\end{lemma}
\begin{proof}
The proof is by contradiction.
Suppose there exists a set $\mat{S}$ of
$k \geq \frac{4}{5}K$ vectors from $\mat{V}$ such that $\vol(\mat{S}) \geq 2^{-\beta^{\circ} k}$.

Consider a labeling $\sigma : (X \uplus Y) \to \Sigma$ defined as follows:
\begin{align}\label{eq:proof-label}
    \sigma(z) \triangleq
    \begin{cases}
    \text{any } i \text{ such that } \vec{v}_{z,i} \in \mat{S}_X & \text{if } z \in X(\mat{S}), \\
    \text{any } i \text{ such that } \vec{v}_{z,i} \in \mat{S}_Y & \text{if } z \in Y(\mat{S}), \\
    \text{any element of } \Sigma & \text{otherwise}.
    \end{cases}
\end{align}
The choice of $i$'s can be arbitrary.
Define $\mat{P} \triangleq \{ \vec{v}_{x,\sigma(x)} \mid x \in X(\mat{S}) \}$ and 
$\mat{Q} \triangleq \{ \vec{v}_{y,\sigma(y)} \mid y \in Y(\mat{S}) \}$.
Our aim is to show that 
the volume of $\mat{P} \uplus \mat{Q}\subseteq \mat{S}$ is sufficiently small.
To use \cref{lem:civril-union},
we bound the distance of the vectors of $\mat{P}$ to $\mat{Q}$.
Since $\mat{Q}$ forms an orthonormal basis by construction and
it holds that
\begin{align*}
    \|\vec{v}_{x,\sigma(x)}\|^2 = \|\proj_{\mat{Q}}(\vec{v}_{x,\sigma(x)})\|^2 + \dis(\vec{v}_{x,\sigma(x)}, \mat{Q})^2
\end{align*}
for each $x \in X(\mat{S})$, we have
\begin{align*}
    \dis(\vec{v}_{x,\sigma(x)}, \mat{Q}) = \sqrt{1 - \sum_{\vec{v}_{y,\sigma(y)} \in \mat{Q}} \langle \vec{v}_{x,\sigma(x)}, \vec{v}_{y,\sigma(y)} \rangle^2}.
\end{align*}
If an edge $(x,y) \in E$ between $X(\mat{S})$ and $Y(\mat{S})$ is not satisfied by $\sigma$; i.e., $\pi_{(x,y)}(\sigma(x)) \neq \sigma(y)$,
then we have
\begin{align*}
    \langle \vec{v}_{x,\sigma(x)}, \vec{v}_{y,\sigma(y)} \rangle
    = \left\langle \frac{\overline{\vec{b}_{\pi_{(x,y)}(\sigma(x))}}}{15^{\ell/2}}, \frac{\vec{b}_{\sigma(y)}}{15^{\ell/2}} \right\rangle
    = \frac{1}{2 \cdot 15^{\ell}}.
\end{align*}
Consequently, we obtain
\begin{align*}
    \dis(\vec{v}_{x,\sigma(x)}, \mat{Q}) = \left(1 - \frac{U(x)}{4 \cdot 15^{2\ell}} \right)^{\frac{1}{2}},
\end{align*}
where $U(x)$ is defined as the number of unsatisfied edges between $x$ and $Y(\mat{S})$.
Using \cref{lem:civril-union} and the fact that 
$\vol(\mat{Q}) \leq 1$, we have
\begin{align}
    \vol(\mat{P} \uplus \mat{Q})
    & \leq \vol(\mat{Q}) \cdot \prod_{x \in X(\mat{S})} \dis(\vec{v}_{x,\sigma(x)}, \mat{Q}) \notag \\
    & \leq \left(\prod_{x \in X(\mat{S})}
    \left(1-\frac{U(x)}{4 \cdot 15^{2 \ell}} \right) \right)^{\frac{1}{2}} \notag \\
    & \leq 
    \left(\frac{1}{|X(\mat{S})|} \sum_{x \in X(\mat{S})} \left(1-\frac{U(x)}{4 \cdot 15^{2\ell}}\right) \right)^{\frac{|X(\mat{S})|}{2}}, \label{eq:vol-unsat}
\end{align}
where the last inequality is by the AM--GM inequality.

Now consider bounding $\sum_{x \in X(\mat{S})} U(x)$ from below, which is equal to
the total number of unsatisfied edges
between $X(\mat{S})$ and $Y(\mat{S})$ by $\sigma$.
Substituting $\beta^{\circ}$ for $c$ in \cref{clm:S-bound} derives
$|X(\mat{S})| > (\frac{3}{8}-10\beta^{\circ})k$ and
$|Y(\mat{S})| > (\frac{3}{8}-10\beta^{\circ})k$.
Because less than $2^{-\alpha \ell}$-fraction of edges in $E$ can be satisfied by
any labeling (including $\sigma$) by assumption, and
more than $( \frac{3}{8}-10 \beta^{\circ} )k \cdot 15^{\ell} $ edges are incident to $X(\mat{S})$ (resp.~$Y(\mat{S})$),
the number of unsatisfied edges \emph{incident to} $X(\mat{S})$ (resp.~$Y(\mat{S})$) is at least 
\begin{align}\label{eq:unsat-incident}
\left( \frac{3}{8} - 10\beta^{\circ} \right)k\cdot 15^\ell - 2^{-\alpha \ell}\cdot (75n)^\ell = 
    \left[ \left(\frac{3}{8}-10\beta^{\circ}\right)\frac{k}{(5n)^{\ell}} - 2^{-\alpha \ell} \right] (75n)^{\ell}.
\end{align}
Consequently, the number of unsatisfied edges
\emph{between} $X(\mat{S})$ and $Y(\mat{S})$ is at least
twice \cref{eq:unsat-incident} minus
``the number of unsatisfied edges incident to $X(\mat{S})$ \emph{or} $Y(\mat{S})$'' (which is at most $(75n)^{\ell}$); namely, we have
\begin{align}\label{eq:unsat-total}
    \sum_{x \in X(\mat{S})} U(x)
    & \geq 2 \cdot (\text{value of \cref{eq:unsat-incident}}) - (75n)^\ell \\
    & = \left[ \left( \frac{3}{8} - 10\beta^{\circ} \right)\frac{2k}{(5n)^\ell} - 1 -2^{-\alpha \ell+1} \right](75n)^\ell \\
    & \geq \left[\left(\frac{1}{16}-10\beta^{\circ}\right)\frac{2k}{(5n)^{\ell}} - 2^{-\alpha \ell+1} \right](75n)^{\ell},
\end{align}
where we have used the fact that $k \geq \frac{4}{5}K$ and $K = 2 \cdot (5n)^{\ell}$.
With this inequality,
we further expand \cref{eq:vol-unsat} as
\begin{align*}
    \vol(\mat{P} \uplus \mat{Q}) & \leq \left(1-\frac{\sum_{x \in X(\mat{S})} U(x)}{|X(\mat{S})|}\frac{1}{4 \cdot 15^{2\ell}} \right)^{\frac{|X(\mat{S})|}{2}} \\
    & \leq \exp\left(-
    \frac{
    \left[ \left(\frac{1}{16}-10 \beta^{\circ}\right) \frac{2k}{(5n)^{\ell}} - 2^{-\alpha \ell + 1} \right](75n)^{\ell}
    }{|X(\mat{S})|}
    \frac{|X(\mat{S})|}{8 \cdot 15^{2 \ell}}
    \right)\\
    & \leq \exp\left(-\left[
    \left(\frac{1}{16} - 10 \beta^{\circ} \right)\frac{1}{4 \cdot 15^{\ell}} - 2^{-\alpha \ell + 1} \frac{5}{64 \cdot 15^{\ell}}
    \right]k\right) \\
    & = \exp\left(-\frac{1-160\beta^{\circ} -5 \cdot 2^{-\alpha \ell + 1}}{64 \cdot 15^{\ell}}k\right).
\end{align*}

Since $\beta^{\circ} = 10^{-10^{12.4}} < \frac{1-5 \cdot 2^{-\alpha \ell + 1}}{64 \cdot 15^{\ell} \cdot \log_{\rme}(2) + 160}$
(recall that $\ell = \lceil \frac{4}{\alpha} \rceil$) and $\beta^{\circ} > 0$,
we finally have
$\vol(\mat{S}) \leq \vol(\mat{P} \uplus \mat{Q}) < 2^{-\beta^{\circ}k}$, a contradiction.
\end{proof}

\begin{proof}[Proof of \nameref{lem:main}]
Let $\GG^{\otimes \ell}$ be the $\ell$-fold parallel repetition
of a special projection game, and
let $\mat{V}$ be an instance of \volmax
reduced from $\GG^{\otimes \ell}$ according to the procedure described in the beginning of this subsection.
Then,
the completeness follows from \cref{lem:completeness} and
the soundness follows from \cref{lem:soundness-1,lem:soundness-2}.
\end{proof}

\section{Constant-Factor Inapproximability for Log-Determinant Maximization}\label{sec:log}

Here, we present constant-factor inapproximability results for log-determinant maximization.
Let us begin with the definition of two optimization problems:
Given a positive semi-definite matrix $\mat{A}$ in $\bbQ^{n \times n}$, 
\emph{log-determinant maximization} (\logdetmax) requests to find a subset $S \subseteq [n]$ that maximizes $\log \det(\mat{A}_S)$, and
\emph{size-constrained log-determinant maximization} ($k$-\logdetmax) requests to find a size-$k$ subset $S \in {[n] \choose k}$ that maximizes $\log\det(\mat{A}_S)$.
Denote $\maxdet_k(\mat{A}) \triangleq \max_{S \in {[n] \choose k}} \det(\mat{A}_S)$ and
$\maxvol_k(\mat{V}) \triangleq \max_{\mat{S} \in {\mat{V} \choose k}} \vol(\mat{S})$ for a vector set $\mat{V}$.

Our first result is a $\frac{5}{4}$-factor inapproximability for \logdetmax,
which is an immediate consequence of the proof of \cref{thm:subdetmax-inapprox}.
\begin{theorem}
\label{thm:log:inapprox}
It is \NP-hard to approximate \logdetmax within a factor of $\frac{5}{4}$.
\end{theorem}
\begin{proof}
The proof of \cref{thm:subdetmax-inapprox} implies that determining whether
$\log_2 \maxdet(\mat{A}) \geq \lambda_c n$ or
$\log_2 \maxdet(\mat{A}) < \lambda_s n$
is \NP-hard.
Observing that $\frac{\lambda_c}{\lambda_s} = \frac{5}{4}$ is sufficient to complete the proof.
\end{proof}

Our next result states that $k$-\logdetmax is inapproximable within a constant factor even if the log-determinant function is nonnegative, monotone, and submodular.
\begin{theorem}
\label{thm:log:size-inapprox}
It is \NP-hard to approximate $k$-\logdetmax within a factor of $1 + 10^{-10^{13}}$,
even if the minimum eigenvalue of an input matrix $\mat{A}$ is at least $1$; namely,
the set function $f(S) \triangleq \log \det(\mat{A}_S)$ is nonnegative, monotone, and submodular.
\end{theorem}

The following lemma is essential for proving \cref{thm:log:size-inapprox}.
\begin{lemma}
\label{lem:log:main}
There is a polynomial-time reduction from the $\ell$-fold parallel repetition $\GG^{\otimes \ell}$ of a special projection game to a set $\mat{V} = \{\vec{v}_1, \ldots, \vec{v}_N\}$ of $N$ vectors such that
$N = 2\cdot (35n)^\ell$ for some integer $n$,
each vector of $\mat{V}$ is normalized, and the following conditions are satisfied:
\begin{itemize}
    \item (Completeness) If $\val(\GG^{\otimes \ell}) = 1$, then there exists a set $\mat{S}$ of $K$ vectors from $\mat{V}$ such that
    $\vol(\mat{S}) = 1$, where $K = \frac{N}{7^\ell}$.
    \item (Soundness) If $\val(\GG^{\otimes \ell}) < 2^{-\alpha \ell}$, then any set $\mat{S}$ of $K$ vectors from $\mat{V}$ satisfies that $\vol(\mat{S}) < 2^{-\beta^{\circ} K}$, where
    $\beta^{\circ} = 10^{-10^{12.4}}$.
    \item (Minimum eigenvalue) The minimum eigenvalue of the Gram matrix defined from $\mat{V}$ is at least $\frac{1}{15^\ell+1}$.
\end{itemize}
\end{lemma}
By \cref{lem:log:main}, we can prove \cref{thm:log:size-inapprox} as follows:

\begin{proof}[Proof of \cref{thm:log:size-inapprox}]
Let $\mat{V}$ be a set of $N$ vectors
reduced from the $\ell$-fold parallel repetition $\GG^{\otimes \ell}$ by \cref{lem:log:main}.
Create a new vector set
$\mat{W} = \{\vec{w}_1, \ldots, \vec{w}_N\}$, where
$\vec{w}_i = (15^\ell+1)^{1/2} \cdot \vec{v}_i$ for each $i \in [N]$.
If $\val(\GG^{\otimes \ell}) = 1$, then there exists a set $\mat{S}$ of $K = N/7^{\ell}$ vectors from $\mat{W}$ such that
$\vol(\mat{S}) = (15^\ell+1)^{K/2}$.
On the other hand, if $\val(\GG^{\otimes \ell}) < 2^{-\alpha \ell}$,
then $\maxvol_k(\mat{W})$ is (strictly) bounded from above by
$2^{-\beta^{\circ}K} \cdot (15^\ell+1)^{K/2}$.
By \cref{thm:dinur-steurer,thm:inapprox-labelcover},
it is \NP-hard to decide whether 
$\maxvol_k(\mat{W}) \geq (15^\ell+1)^{K/2}$ or
$\maxvol_k(\mat{W}) < 2^{-\beta^{\circ}K} \cdot (15^\ell+1)^{K/2}$.

Consider a Gram matrix $\mat{A}$ derived from $\mat{W}$ and
a submodular set function $f: 2^{[N]} \to \bbR$ such that
$f(S) \triangleq \log_2\det(\mat{A}_S)$ for each $S \subseteq [N]$.
Since $\mat{A}$ has a minimum eigenvalue at least $1$ by assumption,
$f$ is monotone \citep[Proposition 2]{sharma2015greedy}, and thus is nonnegative.
Since it is \NP-hard to decide whether
$\log_2 \maxdet_k(\mat{A}) \geq \log_2((15^\ell+1)^K)$ or
$\log_2 \maxdet_k(\mat{A}) < \log_2(2^{-2\beta^{\circ}K }\cdot (15^\ell+1)^K)$,
it is also \NP-hard to approximate $k$-\logdetmax within a factor of 
\begin{align*}
    \frac{\log_2 ((15^\ell+1)^K)}{\log_2 (2^{-2\beta^{\circ}K}\cdot (15^\ell+1)^K)} = \frac{\log_2 (15^\ell+1)}{\log_2(15^\ell+1) - 2\beta^{\circ}} > 1+10^{-10^{13}}.
\end{align*}
This completes the proof.
\end{proof}

The remainder of this section is devoted to the proof of \cref{lem:log:main}.
\subsection{Proof of \cref{lem:log:main}}

Again, the proof of \cref{lem:log:main} relies on a reduction from the $\ell$-fold parallel repetition of a special projection game.
One might think that the vector set $\mat{V}$ introduced in
\cref{subsec:main-lemma} can be used as it is; however, we have the following issue:
Suppose there exist a vertex $x \in X$ and two labels $i_1,i_2 \in \Sigma$ such that
$\pi_e(i_1) = \pi_e(i_2)$ for every edge $e$ incident to $x$.
Then, we have $\vec{v}_{x,i_i} = \vec{v}_{x,i_2}$; i.e., $\mat{V}$ is \emph{linearly dependent}.
It then turns out that for a Gram matrix $\mat{A}$ defined from $\mat{V}$,
$\det(\mat{A}_{\emptyset}) = 1$ and
$\det(\mat{A}) = 0$, implying that
the set function $f_c(S) \triangleq c \log \det(\mat{A}_S) $ is never nonnegative nor monotone for any $c>0$.

To circumvent this issue, we ``modify'' the reduction as follows.
Let $\GG^{\otimes \ell} = (X,Y,E,\Sigma,\Pi)$ be
the $\ell$-fold parallel repetition of a special projection game.
Recall that
$(X,Y,E)$ is a $15^{\ell}$-regular bipartite graph, where
$|X|=|Y|=(5n)^{\ell}$,
$|E|=(75n)^{\ell}$, and
$|\Sigma|=7^{\ell}$ for some integer $n$.
Assume that $\Sigma = [7^{\ell}]$.

For each pair of a vertex of $X \uplus Y$ and a label of $\Sigma$,
we define a vector as follows.
Each vector consists of $|E|$ blocks, each of which is $2^{3\ell+1}$-dimensional and is either a vector in the set
$\mat{B}^{(3\ell)} = \{ \vec{b}_1, \ldots, \vec{b}_{2^{3\ell}} \}$ or the zero vector $\vec{0}$,
and consists of $|X \uplus Y|\cdot |\Sigma|$ entries, each of which is a scalar.
Let $\vec{v}_{x,i}$ (resp.~$\vec{v}_{y,i}$)
denote the vector for a pair $(x, i) \in X \times \Sigma$
(resp.~$(y,i) \in Y \times \Sigma$), and
let $ \vec{v}_{x,i}(e)$ (resp.~$\vec{v}_{y,i}(e)$) denote
the block of $\vec{v}_{x,i}$ (resp.~$\vec{v}_{y,i}$)
corresponding to edge $e \in E$,
let $v_{x,i}(z,j)$ (resp.~$v_{y,i}(z,j)$) denote
the entry of $\vec{v}_{x,i}$ (resp.~$\vec{v}_{y,i}$)
corresponding to pair $(z,j) \in (X \uplus Y) \times \Sigma$.
Hence, each vector is $((75n)^{\ell} \cdot 2^{3\ell+1} + 2 \cdot (35n)^{\ell})$-dimensional.
Each block and entry is defined as follows:
\begin{align*}
    \vec{v}_{x,i}(e) & \triangleq
    \begin{cases}
    \displaystyle\frac{\overline{\vec{b}_{\pi_e(i)}}}{(15^\ell + 1)^{1/2}} & \text{if } e \text{ is incident to } x,  \\
    \vec{0} & \text{otherwise},
    \end{cases}
    & v_{x,i}(z,j) & \triangleq
    \begin{cases}
    \displaystyle\frac{1}{(15^\ell + 1)^{1/2}} & \text{if } (z,j) = (x,i),  \\
    0 & \text{otherwise},
    \end{cases} \\
    \vec{v}_{y,i}(e) & \triangleq
    \begin{cases}
    \displaystyle\frac{\vec{b}_i}{(15^\ell + 1)^{1/2}} & \text{if } e \text{ is incident to } y, \\
    \vec{0} & \text{otherwise},
    \end{cases} 
    & v_{y,i}(z,j) & \triangleq
    \begin{cases}
    \displaystyle\frac{1}{(15^\ell + 1)^{1/2}} & \text{if } (z,j) = (x,i), \\
    0 & \text{otherwise}.
    \end{cases}
\end{align*}
Since each vector contains exactly $15^{\ell}$ blocks
chosen from $\mat{B}^{(3\ell)}$ and
exactly one element whose value is $\frac{1}{(15^\ell+1)^{1/2}}$,
it is normalized;
i.e., $\|\vec{v}_{x,i}\| = \|\vec{v}_{y,i}\| = 1$ for all
$x \in X, y \in Y, i \in \Sigma$.
Note that
$\vec{v}_{x_1, i_1}$ and $\vec{v}_{x_2,i_2}$ (resp.~$\vec{v}_{y_1, i_1}$ and $\vec{v}_{y_2,i_2}$)
are orthogonal for any $x_1,x_2 \in X$ (resp.~$y_1,y_2 \in Y$) and
$i_1,i_2 \in \Sigma$ if $x_1 \neq x_2$ (resp.~$y_1 \neq y_2$), and
$\vec{v}_{x,i}$ and $\vec{v}_{y,j}$ for $x \in X, y \in Y, i,j \in \Sigma$ are orthogonal if
$(x,y) \in E$ and $\pi_{(x,y)}(i) = j$,
or if $ (x,y) \not \in E $.

We then define an instance $\mat{V}$ of \volmax as follows:
\begin{align*}
    \mat{V} \triangleq
    \{ \vec{v}_{x,i} \mid x \in X, i \in \Sigma \} \uplus
    \{ \vec{v}_{y,i} \mid y \in Y, i \in \Sigma \}.
\end{align*}
Here, $\mat{V}$ contains $N \triangleq 2 \cdot (35n)^{\ell} $ vectors.
Define $K \triangleq |X|+|Y| = 2 \cdot (5n)^{\ell}$;
it holds that $K = N / 7^{\ell}$.
Construction of $\mat{V}$ from $\GG^{\otimes \ell}$ can be done in polynomial time in $n$.

\paragraph{Completeness.}
\begin{lemma}
\label{lem:log:completeness}
If $\val(\GG^{\otimes \ell}) = 1$, then there exists a set $\mat{S}$ of $K$ vectors from $\mat{V}$ such that $\vol(\mat{S}) = 1$.
\end{lemma}
\begin{proof}
The proof can be done in a similar manner to that for \cref{lem:completeness}.
\end{proof}

\paragraph{Soundness.}
For a vector set $\mat{S} \subseteq \mat{V}$,
we use the same notation
$\mat{S}_X$, $\mat{S}_Y$, $X(\mat{S})$, $Y(\mat{S})$, $\rep(\mat{S}_X)$, and $\rep(\mat{S}_Y)$ as those introduced in \cref{subsec:proof-lemma}.
We first prove the following, whose proof is similar to that for \cref{lem:civril-vol} and deferred to \cref{app:proofs}.
\begin{lemma}
\label{lem:log:civril-vol}
For any set $\mat{S}$ of vectors from $\mat{V}$, it holds that
\begin{align*}
    \vol(\mat{S}_X) \leq \left(\frac{\sqrt{3.01}}{2}\right)^{\rep(\mat{S}_X)} \text{ and }
    \vol(\mat{S}_Y) \leq \left(\frac{\sqrt{3.01}}{2}\right)^{\rep(\mat{S}_Y)}.
\end{align*}
\end{lemma}

By using \cref{lem:log:civril-vol},
we show that both $X(\mat{S})$ and $Y(\mat{S})$ contain $\Omega(K)$ vertices if its volume is sufficiently large, whose proof is similar to that for \cref{clm:S-bound} and deferred to \cref{app:proofs}.
\begin{claim}
\label{clm:log:S-bound}
For any set $\mat{S}$ of $K$ vectors from $\mat{V}$,
if $\vol(\mat{S}) \geq 2^{-cK}$ for some number $c > 0$,
then it holds that
\begin{align*}
    |X(\mat{S})| > \left(\frac{1}{2} - 10c\right)K \text{ and }
    |Y(\mat{S})| > \left(\frac{1}{2} - 10c\right)K.
\end{align*}
\end{claim}

We now show that no vector set of size $K$ has a volume close to $1$ if $\val(\GG^{\otimes \ell})$ is small,
whose proof is similar to that for \cref{lem:soundness-2} and deferred to \cref{app:proofs}.
\begin{lemma}
\label{lem:log:soundness-2}
Suppose $\val(\GG^{\otimes \ell}) < 2^{-\alpha \ell}$.
For any set $\mat{S}$ of $K$ vectors from $V$,
it holds that $\vol(\mat{S}) < 2^{- \beta^{\circ} K}$, where 
$\beta^{\circ} = 10^{-10^{12.4}}$.
\end{lemma}

\paragraph{Eigenvalue.}
We finally bound the minimum eigenvalue of the Gram matrix from below.
\begin{lemma}
\label{lem:log:eigen}
The minimum eigenvalue of a Gram matrix defined from $\mat{V}$ is at least $\frac{1}{15^\ell + 1}$.
\end{lemma}
\begin{proof}
Let $\mat{A} \in \bbR^{N \times N}$ be a Gram matrix defined from $\mat{V}$; i.e.,
it holds that $A_{i,j} \triangleq \langle \vec{v}_i, \vec{v}_j \rangle$ for all $i,j \in [N]$.
By the Courant--Fischer min-max theorem, the minimum eigenvalue of $\mat{A}$, denoted $\lambda_1(\mat{A})$, is equal to the minimum possible value of the Rayleigh quotient; namely,
\begin{align*}
    \lambda_1(\mat{A}) = \min_{\|\vec{b}\| = 1} \vec{b}^\top \mat{A} \vec{b}.
\end{align*}
Recall that each vector $\vec{v}$ of $\mat{V}$ consists of
blocks $\vec{v}(e)$ corresponding to edge $e \in E$ and
entries $v(z,j)$ corresponding to pair $(z,j) \in (X \uplus Y) \times \Sigma$.
Suppose an $N$-dimensional vector $\vec{b}$ is indexed by pairs of $(X \uplus Y) \times \Sigma$,
that is, let $b(z,i)$ denote an entry of $\vec{b}$ corresponding to pair $(z,i) \in (X\uplus Y) \times \Sigma$.
By definition of vectors of $\vec{V}$, we have
\begin{align*}
    \vec{b}^\top \mat{A} \vec{b}
    & \geq \sum_{(z,j) \in (X \uplus Y) \times \Sigma}
    \left( \sum_{(x,i) \in X \times \Sigma} b(x,i) \cdot v_{x,i}(z,j) + \sum_{(y,i) \in Y \times \Sigma} b(y,i) \cdot v_{y,i}(z,j) \right)^2 \\
    & \geq \sum_{(z,j) \in (X \uplus Y) \times \Sigma} (b(z,j) \cdot v_{z,j}(z,j))^2 \\
    & = \frac{1}{15^\ell + 1} \sum_{(z,j) \in (X \uplus Y) \times \Sigma} b(z,j)^2 = \frac{1}{15^\ell + 1} \|\vec{b}\|^2.
\end{align*}
Consequently, we have
$ \displaystyle\lambda_1(\mat{A}) \geq \min_{\|\vec{b}\| = 1} \frac{1}{15^\ell+1}\|\vec{b}\|^2 = \frac{1}{15^\ell+1}$,
which finishes the proof.
\end{proof}

\begin{proof}[Proof of \cref{lem:log:main}]
Let $\GG^{\otimes \ell}$ be the $\ell$-fold parallel repetition of a special projection game, and let $\mat{V}$ be a vector set reduced by the procedure described in the beginning of this section.
Then, the completeness follows from \cref{lem:log:completeness},
the soundness follows from \cref{lem:log:soundness-2}, and
the minimum eigenvalue bound follows from \cref{lem:log:eigen}.
\end{proof}

\section{Exponential Inapproximability for Exponentiated DPPs}\label{sec:e-dpp}

We derive the exponential inapproximability of exponentiated DPPs.
Given an $n \times n$ positive semi-definite matrix $\mat{A}$,
the \emph{exponentiated DPP (E-DPP)} of exponent $p > 0$ defines
a distribution over the power set $2^{[n]}$,
whose probability mass for each subset $S \subseteq [n]$ is 
proportional to $\det(\mat{A}_S)^p$.
We use $\ZZ^p(\mat{A})$ to denote the normalizing constant; namely,
\begin{align*}
\ZZ^p(\mat{A}) \triangleq \sum_{S \subseteq [n]} \det(\mat{A}_S)^p.
\end{align*}
$\ZZ^p(\mat{A})$ must be at least $1$ by the positive semi-definiteness of $\mat{A}$.
We say that an estimate $\widehat{\ZZ}^p$ is a \emph{$\rho$-approximation} to $\ZZ^p$ for $\rho \geq 1$ if it holds that
\begin{align*}
    \ZZ^p \leq \widehat{\ZZ}^p \leq \rho \cdot \ZZ^p.
\end{align*}
For two probability distributions $\bm{\mu}$ and $\bm{\eta}$ on $\Omega$,
the \emph{total variation distance} is defined as
\begin{align*}
\frac{1}{2} \sum_{x \in \Omega}|\mu_x - \eta_x|
\end{align*}

We first present the following theorem stating that
assuming exponential inapproximability of \subdetmax,
we can neither
estimate $\ZZ^p$ and thus the probability mass for any subset accurately
nor generate a random sample from E-DPPs in polynomial time
for a sufficiently large $p$.

\begin{theorem}\label{thm:e-dpp}
Suppose there exist universal constants
$\ccc$ and $\sss$ such that
\gapsubdetmax{$[2^{\sss n}, 2^{\ccc n}]$} is \NP-hard.
Then, for every fixed number $p > \frac{1}{\ccc-\sss}$,
it is \NP-hard to approximate $\ZZ^p(\mat{A})$
for a positive semi-definite matrix $\mat{A}$ in $\bbQ^{n \times n}$
within a factor of $2^{((\ccc-\sss)p -1)n}$.
Moreover, unless \RP~$=$~\NP,\footnote{\RP is the class of decision problems for which
there exists a probabilistic polynomial-time Turing machine that
accepts a \emph{yes} instance with probability $\geq \frac{1}{2}$ and
always rejects a \emph{no} instance.
It is believed that \RP~$\neq$~\NP.}
no polynomial-time algorithm can draw a sample from
a distribution whose total variation distance from
E-DPPs of exponent
$p > \frac{1}{\ccc-\sss}$
is at most $\frac{1}{3}$.
\end{theorem}

As a corollary of \cref{thm:subdetmax-inapprox,thm:e-dpp},
we have the following inapproximability result on E-DPPs.
\begin{corollary}
\label{cor:e-dpp}
For every fixed number $p \geq \beta^{-1} = 10^{10^{13}}$,
it is \NP-hard to approximate $\ZZ^p(\mat{A})$ for an $n \times n$ positive semi-definite matrix $\mat{A}$
within a factor of $2^{\beta p n}$.
Moreover, no polynomial-time algorithm can draw a sample whose total variation distance
from the E-DPP of exponent $p$ is at most $\frac{1}{3}$,
unless \RP~$=$~\NP.
\end{corollary}
\begin{proof}
The proof of \cref{thm:subdetmax-inapprox} implies that
$\ccc = \frac{2 \beta^{\circ}}{7^{\ell}}$ and $\sss = \frac{8 \beta^{\circ}}{5 \cdot 7^{\ell}}$ satisfy the conditions in \cref{thm:e-dpp} and
that $10^{10^{12.7}} > \frac{1}{\ccc-\sss}$.
Observe then that
$(\ccc - \sss)p - 1 > \beta p$ for
$\beta = 10^{-10^{13}}$ and $p \geq 10^{10^{13}}$,
which suffices to complete the proof.
\end{proof}

\begin{proof}[Proof of \cref{thm:e-dpp}]
Consider the E-DPP of exponent $p > \frac{1}{\ccc-\sss}$ defined by
a positive semi-definite matrix $\mat{A} \in \bbQ^{n \times n}$.
Suppose $p = \frac{q}{\ccc-\sss}$ for some $q > 1$.
We prove the first argument.
If there exists a set $S \subseteq [n]$ such that $\det(\mat{A}_S) \geq 2^{\ccc n}$,
then $\ZZ^p(\mat{A})$ is at least $2^{\ccc p n}$.
On the other hand,
if every set $S \subseteq [n]$ satisfies
$\det(\mat{A}_S) < 2^{\sss n}$,
then $\ZZ^p(\mat{A})$ is less than $2^{\sss p n + n}$.
If a $2^{(q-1)n}$-approximation to $\ZZ^p(\mat{A})$ is given,
we can distinguish the two cases
(i.e., we can solve \gapsubdetmax{$[2^{\sss n}, 2^{\ccc n}]$})
because
$2^{(q-1)n} = 2^{\ccc p n}/2^{\sss p n + n}$.
We then prove the second argument.
Assume that $\maxdet(\mat{A}) \geq 2^{\ccc n}$.
Sampling $S$ from the E-DPP,
we have ``$\det(\mat{A}_S) > 2^{\sss n}$''
(which is a certificate of the case)
with probability at least
$\frac{2^{\ccc p n}}{\ZZ^p(\mat{A})}$, and
we have ``$\det(\mat{A}_S) \leq 2^{\sss n}$'' with probability at most
$\frac{2^{\sss pn + n}}{\ZZ^p(\mat{A})}$.
Hence, provided a polynomial-time algorithm to generate a random sample whose total variation distance from the E-DPP is at most $\frac{1}{3}$,
we can use it to find the certificate with probability at least
\begin{align*}
    \left(1 - \frac{2^{\sss p n + n}}{2^{\sss p n + n} + 2^{\ccc pn}}\right) - \frac{1}{3} \geq \frac{2}{3} - \frac{1}{1 + 2^{qn - n}} \geq \frac{1}{2}
    \text{\quad(as long as } n \geq \frac{3}{q-1} \text{)},
\end{align*}
implying \RP~$=$~\NP. This completes the proof.
\end{proof}

\begin{remark}
\label{remark:edpp}
\cref{thm:e-dpp} holds even when
$\ccc$ and $\sss$ are functions in $n$.
If we apply a $(\frac{9}{8}-\epsilon)$-factor inapproximability due to \citet{kulesza2012determinantal},
then we would obtain $\frac{1}{\ccc- \sss} = \Theta(n)$; thus, $p$ must be $ \Omega(n)$, which is weaker than \cref{cor:e-dpp}.
\cref{thm:subdetmax-inapprox}
is crucial for ruling out approximability for constant $p$.
\end{remark}

We finally observe that a $2^{\bigO(pn)}$-approximation to $\ZZ^p$
can be derived using a $2^{\bigO(n)}$-approximation algorithm for \subdetmax \citep{nikolov2015randomized}.
This means that \cref{cor:e-dpp} is tight \emph{up to} a constant in the exponent (when $p \geq 10^{10^{13}}$).

\begin{observation}
\label{obs:edpp-approx}
There exists a polynomial-time algorithm that approximates
$\ZZ^p(\mat{A})$ for an $n \times n$ positive semi-definite matrix $\mat{A}$ within a factor of 
$(2 \cdot \rme^p)^{n}$.
\end{observation}
\begin{proof}
Fix an $n \times n$ positive semi-definite matrix
$\mat{A}$ in $\bbQ^{n \times n}$.
Let $S \subseteq [n]$ be an $\rme^{n}$-approximation
to \subdetmax obtained by running the algorithm of \citet{nikolov2015randomized} for every $k \in [n]$; i.e., it holds that
\begin{align*}
\rme^{-n} \cdot \maxdet(\mat{A}) \leq \det(\mat{A}_S) \leq \maxdet(\mat{A}).
\end{align*}
Thanks to the positive semi-definiteness of $\mat{A}$,
it is easily verified that
\begin{align*}
    \frac{1}{2^n} \cdot \ZZ^p(\mat{A}) \leq \maxdet(\mat{A})^p \leq
    \ZZ^p(\mat{A}).
\end{align*}
Combining two inequalities, we have
\begin{align*}
    \frac{1}{\rho} \cdot \ZZ^p(\mat{A})
    \leq \det(\mat{A}_S)^p
    \leq \ZZ^p(\mat{A}),
\end{align*}
where $\rho \triangleq 2^n \rme^{pn}$.
Hence, $\rho \cdot \det(\mat{A}_S)^p$ is a $\rho$-approximation of $\ZZ^p(\mat{A})$, completing the proof.
\end{proof}

\section{Concluding Remarks and Open Questions}
We have established three inapproximability results for MAP inference on DPPs and the normalizing constant for exponentiated DPPs.
We conclude this paper with three open questions.

\begin{itemize}
    \item \textbf{Optimal bound for unconstrained MAP inference.}
    The universal constant $\beta = 10^{-10^{13}}$ in \cref{thm:subdetmax-inapprox} seems extremely small
    despite $\rme^{n}$-factor approximability \citep{nikolov2015randomized};
    improving the value of $\beta$ is a potential research direction.
    
    \item \textbf{Complexity class of log-determinant maximization.}
    \cref{thm:log:inapprox,thm:log:size-inapprox} state that \logdetmax is $\bigO(1)$-factor inapproximable, while 
    \logdetmax is known to be in the class \APX.\footnote{
    \APX is the class of \NP optimization problems that admit constant-factor approximation algorithms.}
    Is \logdetmax \APX-complete?
    
    \item \textbf{Smallest exponent $p$ for which $\ZZ^p$ is inapproximable.}
    Our upper bound $10^{10^{13}}$ on the exponent $p$
    in \cref{cor:e-dpp} is surprisingly large.
    Can we find a (smaller) ``threshold'' $p_c$ such that
    $\ZZ^p$ is approximable if $p \leq p_c$ and inapproximable otherwise?
\end{itemize}

\acks{This work was mostly done while the author was at NEC Corporation.
The author would like to thank Leonid Gurvits for pointing out to us the reference \citep{gurvits2009complexity}, and
thank the referees for helpful suggestions on the presentation of the paper.}

%% file: tab-subdetmax.tex
\begin{table*}[tbp]
    \centering
    \caption{Computational complexity of MAP inference on DPPs, i.e., $\max_{S} \det(\mat{A}_S)$.
    Our result is $2^{\beta n}$-factor inapproximability,
    improving the known lower bound of $\approx \frac{9}{8}$ and    
     matching the best upper bound of $\rme^{n}$.}
    \label{tab:subdetmax}
    \setlength{\tabcolsep}{2pt}
    \begin{tabular}{c|c|c}
    \toprule
    \textbf{constrained?}  & \textbf{inapproximability} & \textbf{approximability} \\
    \midrule
    unconstrained &
    $2^{\beta n}$ ($\beta = 10^{-10^{13}}$) {\scriptsize(\textbf{this paper, \cref{thm:subdetmax-inapprox}})} &
    $\rme^n$ {\scriptsize\citep{nikolov2015randomized}} \\
    $S \subseteq [n]$ &
    $\frac{9}{8} - \epsilon$ {\scriptsize\citep{kulesza2012determinantal}} &
    $n!^2$ {\scriptsize\citep{civril2009selecting}} \\
    \midrule
    
    size-constrained &
    $2^{ck}$ {\scriptsize \citep{koutis2006parameterized,civril2013exponential}} &
    $\rme^k$ {\scriptsize\citep{nikolov2015randomized}}
    \\
    $S \subseteq [n], |S|=k $ & $(2^{\frac{1}{506}} - \epsilon)^k$ {\scriptsize\citep{summa2014largest}} &
    $k!^2$ {\scriptsize\citep{civril2009selecting}}
    \\
    \bottomrule
    \end{tabular}
\end{table*}

%% file: app.tex
\section{Proof of \cref{thm:inapprox-labelcover}}\label{app:inapprox-labelcover}

The proof of \cref{thm:inapprox-labelcover} is based on
a series of gap-preserving reductions from \prb{Max-3SAT} to \prb{LabelCover}.
In \prb{Max-3SAT},
given a 3-conjunctive normal form (3-CNF) Boolean formula $\phi$,
of which each clause contains at most three variables
(e.g., $\phi = (x_1 \vee \overline{x_2} \vee x_3) \wedge (x_2 \vee \overline{x_3} \vee x_4) \wedge (\overline{x_1} \vee x_3 \vee \overline{x_4} )$),
we are asked to find a truth assignment that satisfies
the maximum fraction of the clauses of $\phi$.
The decision version of \prb{Max-3SAT} is known to be \NP-hard;
indeed, \citet*{hastad2001some} established the following indistinguishability:

\begin{theorem}[\protect{\citealp[Theorem 6.5]{hastad2001some}}]
\label{thm:max3sat}
Given a 3-CNF Boolean formula $\phi$,
it is \NP-hard to distinguish between the following two cases
for any constant $\epsilon>0$:
\begin{itemize}
    \item (Completeness) $\phi$ is satisfiable.
    \item (Soundness) No truth assignment can satisfy at least
    $(\frac{7}{8} + \epsilon)$-fraction of the clauses of $\phi$.
\end{itemize}
\end{theorem}

We first reduce from \prb{Max-3SAT} to \prb{Max-3SAT(29)},
which is a special case of \prb{Max-3SAT} where
every variable appears in at most $29$ clauses of a 3-CNF Boolean formula.
\begin{theorem}[%
\protect{\citealp[Theorem 7]{trevisan2004inapproximability}};
\protect{\citealp[Theorem 29.11]{vazirani2013approximation}};
\protect{\citealp[Theorem 4.7]{tamaki2015parallel}}%
]
\label{thm:max3sat29}
There is a polynomial-time gap-preserving reduction that
transforms an instance $\phi$ of \prb{Max-3SAT} to
an instance $\psi$ of \prb{Max-3SAT(29)} such that
the following is satisfied:
\begin{itemize}
    \item (Completeness) If $\phi$ is satisfiable, then so is $\psi$.
    \item (Soundness) If no truth assignment satisfies at least $(1-\epsilon)$ faction of the clauses of $\phi$, then no truth assignment satisfies at least $(1-\frac{\epsilon}{43})$-fraction of the clauses of $\psi$.
\end{itemize}
\end{theorem}

We next reduce from \prb{Max-3SAT(29)} to \prb{Max-E3SAT(5)},
which is a special case of \prb{Max-3SAT} where
an input 3-CNF formula contains $n$ variables and $5n/3$ clauses for some positive integer divisible by $3$,
each clause contains exactly $3$ literals, and
each variable appears in exactly $5$ clauses but
never appears twice in the same clause.

\begin{theorem}[\protect{\citealp[Proposition 2.1.2]{feige1998threshold}}; \protect{\citealp[Theorem 4.11]{tamaki2015parallel}}]
\label{thm:maxe3sat5}
There is a polynomial-time gap-preserving reduction that
transforms
an instance $\phi$ of \prb{Max-3SAT(29)} to 
an instance $\psi$ of \prb{Max-E3SAT(5)} such that
the following is satisfied:
\begin{itemize}
    \item (Completeness) If $\phi$ is satisfiable, then so is $\psi$.
    \item (Soundness) If no truth assignment satisfies at least $(1-\epsilon)$-fraction of the clauses of $\phi$, then
    no truth assignment satisfies at least $(1-\frac{\epsilon}{200})$-fraction of the clauses of $\psi$.
\end{itemize}
\end{theorem}

We further reduce from \prb{Max-E3SAT(5)} to
 \prb{LabelCover} (i.e., a projection game) as follows.

\begin{theorem}[\protect{\citealp[Proof of Theorem 4.2]{tamaki2015parallel}}]
\label{thm:labelcover}
There is a polynomial-time reduction that
transforms an instance $\phi$ of \prb{Max-E3SAT(5)} with $n$ variables and $5n/3$ clauses 
to a projection game $\GG = (X,Y,E,\Sigma,\Pi)$
such that the following is satisfied:
\begin{itemize}
    \item $|X|=n, |Y|=5n/3, |E|=5n, |\Sigma|=7$, and
    each vertex of $X$ and $Y$ has exactly degree $5$ and $3$, respectively.
    \item (Completeness) If $\phi$ is satisfiable, then $\val(\GG)=1$.
    \item (Soundness) If no truth assignment satisfies at least $(1-\epsilon)$-fraction of
    the clauses of $\phi$,
    then $\val(\GG) \leq 1-\frac{\epsilon}{3}$.
\end{itemize}
\end{theorem}

We finally reduce from the above-mentioned projection game to a special projection game (see \cref{subsec:inapprox-game}) as follows.
\begin{lemma}
\label{lem:special-game}
Let $\GG' = (X',Y',E',\Sigma',\Pi' = \{\pi'_e\}_{e \in E'})$ be
a projection game
satisfying the conditions in \cref{thm:labelcover},
i.e., $|X'|=n$, $|Y'|=5n/3$, $|E'|=5n$, and $|\Sigma'|=7$ for some positive integer $n$ divisible by $3$.
There is a polynomial-time gap-preserving reduction that transforms
$\GG'$ to a new projection game
$ \GG = (X,Y,E,\Sigma,\Pi = \{\pi_e\}_{e \in E}) $ such that the following is satisfied:
\begin{itemize}
    \item $(X,Y,E)$ is a $15$-regular bipartite graph
    (i.e., each vertex of $X \uplus Y$ is incident to exactly $15$ edges),
    where $|X| = |Y| = 5n$ and $|E|=75n$,
    \item $ |\Sigma| = 7 $, and
    \item $\val(\GG) = \val(\GG')$.
\end{itemize}
\end{lemma}
\begin{proof}
We construct $\GG$ as follows.
We first create
$X \triangleq \bigcup_{x \in X'} \{x^{(1)}, x^{(2)}, x^{(3)}, x^{(4)}, x^{(5)} \}$ and 
$Y \triangleq \bigcup_{y \in Y'} \{ y^{(1)}, y^{(2)}, y^{(3)} \}$,
where each $x^{(i)}$ for $i \in [5]$ and 
$y^{(j)}$ for $j \in [3]$ is a copy of $x \in X$ and $y \in Y$,
respectively.
We then build an edge set as
$E \triangleq \bigcup_{i \in [5], j \in [3]} \{ (x^{(i)}, y^{(j)}) \in X \times Y \mid (x, y) \in E' \}$, and
$ \pi_{e} \triangleq \pi'_{(x,y)} $ for each $e = (x^{(i)}, y^{(j)}) \in E$.
Then, $\val(\GG)$ is easily verified to be $\val(\GG')$.
\end{proof}

\begin{proof}[Proof of \cref{thm:inapprox-labelcover}]
By \cref{thm:max3sat,thm:max3sat29,thm:maxe3sat5,thm:labelcover,lem:special-game},
we can reduce a 3-CNF Boolean formula $\phi$ to a projection game $\GG$ such that the following is satisfied:
\begin{itemize}
    \item $\GG$ is special (i.e., $\GG$ satisfies the conditions in \cref{thm:inapprox-labelcover}).
    \item If $\phi$ is satisfiable, then $\val(\GG) = 1$.
    \item If no truth assignment satisfy at most $(\frac{7}{8}+\epsilon)$-fraction of the clauses of $\phi$, then
    $\val(\GG) \leq 1-\frac{1}{206{,}400}+\frac{\epsilon}{25{,}800} $,
    for any $\epsilon > 0$.
\end{itemize}
We can specify $\epsilon > 0$ so that
$1-\frac{1}{206{,}400} + \frac{\epsilon}{25{,}800} < 1-\frac{1}{206{,}401}$,
which completes the proof.
\end{proof}

\section{Missing Proofs in \cref{sec:subdetmax} and \cref{sec:log}}\label{app:proofs}

\begin{proof}[Proof of \cref{lem:civril-vol}]
Fix $\mat{S} \subseteq \mat{V}$.
Let $\mat{Q}$ be a set of $|X(\mat{S})|$ vectors from $\mat{S}_X$
with no repetitions, i.e.,
$\mat{Q} \subseteq \mat{S}_X $ such that
$|\mat{Q}| = |X(\mat{S})|$ and
$\rep(\mat{Q}) = 0$, and let $\mat{P} = \mat{S}_X \setminus \mat{Q}$.
For each vector $\vec{v}_{x,i}$ in $\mat{P}$,
there is exactly one vector $\vec{v}_{x,i'}$ in $\mat{Q}$ with $i \neq i'$; hence, we have
\begin{align*}
\dis(\vec{v}_{x,i}, \mat{Q}) \leq \dis(\vec{v}_{x,i}, \{\vec{v}_{x,i'}\}) = \| \vec{v}_{x,i} - \langle \vec{v}_{x,i}, \vec{v}_{x,i'} \rangle \vec{v}_{x,i'} \| \leq \frac{\sqrt{3}}{2}.
\end{align*}
By \cref{lem:civril-union}, we have
\begin{align*}
    \vol(\mat{S}_X) \leq \vol(\mat{Q}) \cdot \prod_{\vec{v}_{x,i} \in \mat{P}} \dis(\vec{v}_{x,i}, \mat{Q}) \leq \left(\frac{\sqrt{3}}{2}\right)^{|\mat{P}|} = \left(\frac{\sqrt{3}}{2}\right)^{\rep(\mat{S}_X)}.
\end{align*}
The proof for $\vol(\mat{S}_Y)$ is similar.
\end{proof}

\begin{proof}[Proof of \cref{lem:log:civril-vol}]
Fix $\mat{S} \subseteq \mat{V}$.
Let $\mat{Q}$ be a set of $|X(\mat{S})|$ vectors from $\mat{S}_{X}$ with no repetitions, i.e.,
$\mat{Q} \subseteq \mat{S}_X$ such that $|\mat{Q}| = |X(\mat{S})|$ and $\rep(\mat{Q})=0$, and
let $\mat{P} = \mat{S}_X \setminus \mat{Q}$.
For each vector $\vec{v}_{x,i}$ in $\mat{P}$, there is exactly one vector $\vec{v}_{x,i'}$ in $\mat{Q}$ with $i \neq i'$; hence, we have
\begin{align*}
    \dis(\vec{v}_{x,i}, \mat{Q}) & \leq \dis(\vec{v}_{x,i}, \{\vec{v}_{x,i'}\})
    = \| \vec{v}_{x,i} - \langle \vec{v}_{x,i}, \vec{v}_{x,i'} \rangle \vec{v}_{x,i'} \| \\
    & = \sqrt{1-\langle \vec{v}_{x,i}, \vec{v}_{x,i'} \rangle^2}
    \leq \sqrt{1- \left(\frac{15^\ell}{15^\ell+1} \cdot \frac{1}{2}\right)^2 } \\
    & < \frac{\sqrt{3.01}}{2}.
\end{align*}
By \cref{lem:civril-union}, we have
\begin{align*}
    \vol(\mat{S}_X) \leq \vol(\mat{Q}) \cdot \prod_{\vec{v}_{x,i} \in \mat{P}} \dis(\vec{v}_{x,i}, \mat{Q}) \leq \left(\frac{\sqrt{3.01}}{2}\right)^{|\mat{P}|} = \left(\frac{\sqrt{3.01}}{2}\right)^{\rep(\mat{S}_X)}.
\end{align*}
The proof for $\vol(\mat{S}_Y)$ is similar.
\end{proof}

\begin{proof}[Proof of \cref{clm:log:S-bound}]
Observe first that $\vol(\mat{S}_X) \geq \vol(\mat{S}) \geq 2^{-cK}$.
By \cref{lem:log:civril-vol}, we have
$
    (\sqrt{3.01}/2)^{\rep(\mat{S}_X)}
    \geq \vol(\mat{S}_X) \geq 2^{-cK},
$
implying $ \rep(\mat{S}_X) \leq cK / \log_2(2/\sqrt{3.01}) < 5cK $. Similarly,
we have $\rep(\mat{S}_Y) < 5cK$.
Using the fact that
$|\mat{S}_X| = |X(\mat{S})| + \rep(\mat{S}_X)$,
$|\mat{S}_Y| = |Y(\mat{S})| + \rep(\mat{S}_Y)$, and
$K = |\mat{S}_X| + |\mat{S}_Y|$,
we bound $|X(\mat{S})|$ from below as follows:
\begin{align*}
    |X(\mat{S})| & = K - |\mat{S}_Y| - \rep(\mat{S}_X)  > K - |Y| - 10cK \\
    & = \left(1 - \frac{(5n)^\ell}{K} - 10c \right)K = \left(\frac{1}{2} - 10c\right)K.
\end{align*}
Similarly, we have
$|Y(\mat{S})| > (\frac{1}{2} - 10c)K$.
\end{proof}

\begin{proof}[Proof of \cref{lem:log:soundness-2}]
The proof is by contradiction.
Suppose there exists a set $\mat{S}$ of $K$ vectors from $V$ such that $\vol(\mat{S}) \geq 2^{-\beta^{\circ} K} $.

Consider a labeling $\sigma : (X \uplus Y) \to \Sigma$ defined according to \cref{eq:proof-label} in the proof of \cref{lem:soundness-2}.
Define $\mat{P} \triangleq \{ \vec{v}_{x,\sigma(x)} \mid x \in X(\mat{S}) \}$ and 
$\mat{Q} \triangleq \{ \vec{v}_{y,\sigma(y)} \mid y \in Y(\mat{S}) \}$.
Here, to use \cref{lem:civril-union},
we bound the distance of the vectors of $\mat{P}$ to $\mat{Q}$.
Similar to the proof of \cref{lem:soundness-2},
we have for each $\vec{v}_{x,\sigma(x)} \in \mat{P}$,
\begin{align*}
    \dis(\vec{v}_{x,\sigma(x)}, \mat{Q}) = \sqrt{1 - \sum_{\vec{v}_{y,\sigma(y)} \in \mat{Q}} \langle \vec{v}_{x,\sigma(x)}, \vec{v}_{y,\sigma(y)} \rangle^2}
    = \left(1 - \frac{U(x)}{4 \cdot (15^{\ell}+1)^2} \right)^{\frac{1}{2}},
\end{align*}
where $U(x)$ is defined as the number of unsatisfied edges between $x$ and $Y(\mat{S})$.
Using \cref{lem:civril-union} and the fact that 
$\vol(\mat{Q}) \leq 1$, we have
\begin{align}
    \vol(\mat{P} \uplus \mat{Q})
    & \leq \vol(\mat{Q}) \cdot \prod_{x \in X(\mat{S})} \dis(\vec{v}_{x,\sigma(x)}, \mat{Q}) \notag \\
    & \leq \left(\prod_{x \in X(\mat{S})}
    \left(1-\frac{U(x)}{4 \cdot (15^{\ell}+1)^2} \right) \right)^{\frac{1}{2}} \notag \\
    & \leq 
    \left(\frac{1}{|X(\mat{S})|} \sum_{x \in X(\mat{S})} \left(1-\frac{U(x)}{4 \cdot (15^{\ell}+1)^2}\right) \right)^{\frac{|X(\mat{S})|}{2}}, \label{eq:log:vol-unsat}
\end{align}
where the last inequality is by the AM--GM inequality.

Now consider bounding $\sum_{x \in X(\mat{S})} U(x)$ from below, which is equal to
the total number of unsatisfied edges between $X(\mat{S})$ and $Y(\mat{S})$ by $\sigma$.
Substituting $\beta^{\circ}$ for $c$ in \cref{clm:log:S-bound} derives
$|X(\mat{S})| > (\frac{1}{2}-10\beta^{\circ})K$ and
$|Y(\mat{S})| > (\frac{1}{2}-10\beta^{\circ})K$.
Because less than $2^{-\alpha \ell}$-fraction of edges in $E$ can be satisfied by any labeling (including $\sigma$) by assumption, and
more than $(\frac{1}{2}-10\beta^{\circ})K$ edges are incident to $X(\mat{S})$ (resp.~$Y(\mat{S})$),
the number of unsatisfied edges \emph{incident to} $X(\mat{S})$ (resp.~$Y(\mat{S})$) is at least
\begin{align}
\label{eq:log:unsat-incident}
    \left[ \left(\frac{1}{2}-10\beta^{\circ}\right)\frac{K}{(5n)^{\ell}} - 2^{-\alpha \ell} \right] (75n)^{\ell}.
\end{align}
Consequently, the number of unsatisfied edges \emph{between} $X(\mat{S})$ and $Y(\mat{S})$ is at least twice \cref{eq:log:unsat-incident} minus
``the number of unsatisfied edges incident to $X(\mat{S})$ \emph{or} $Y(\mat{S})$'' (which is at most $(75n)^{\ell}$); namely, we have
\begin{align*}
    \sum_{x \in X(\mat{S})} U(x) \geq \left[\left(\frac{1}{4}-10\beta^{\circ}\right)\frac{2K}{(5n)^{\ell}} - 2^{-\alpha \ell+1} \right](75n)^{\ell}.
\end{align*}
We further expand \cref{eq:log:vol-unsat} as follows:
\begin{align*}
    \vol(\mat{P} \uplus \mat{Q}) & \leq \left(1 - \frac{\sum_{x \in X(\mat{S})}U(x)}{|X(\mat{S})|}\frac{1}{4 \cdot (15^\ell + 1)^2} \right)^{\frac{|X(\mat{S})|}{2}} \\
    & \leq \exp\left( -\frac{
    \left[ (\frac{1}{4}-10\beta^{\circ}) \frac{2K}{(5n)^{\ell}} - 2^{-\alpha \ell+1} \right](75n)^{\ell}
    }{|X(\mat{S})|}
    \frac{|X(\mat{S})|}{8 \cdot (15^\ell+1)^2} \right) \\
    & \leq \exp \left( -\left[ \left( \frac{1}{4}-10\beta^{\circ} \right)\frac{15^\ell}{4 \cdot (15^\ell+1)^2} - 2^{-\alpha \ell+1}\frac{15^\ell}{16 \cdot (15^\ell+1)^2} \right] K \right) \\
    & = \exp\left(- \frac{1-40\beta^{\circ} -2^{-\alpha \ell+1}}{16 \cdot (15^\ell+1)^2} 15^\ell \cdot K\right) \\
    & \leq \exp\left(- \frac{1-40\beta^{\circ} - 2^{-\alpha \ell +1}}{32 \cdot 15^\ell} K \right).
\end{align*}
Since $\beta^{\circ} = 10^{-10^{12.4}} < \frac{1-2^{-\alpha \ell+1}}{32 \cdot 15^\ell \cdot \log_\rme(2) + 40} $ and $\beta^{\circ} > 0$,
we finally have $\vol(\mat{S}) \leq \vol(\mat{P} \uplus \mat{Q}) < 2^{-\beta^{\circ} K}$, a contradiction.
\end{proof}